\def\draft{1}  
\documentclass[11pt]{article}
\usepackage[letterpaper,hmargin=1in,vmargin=1in]{geometry}

\usepackage{ifthen}

\newcommand{\RJ}[1]{\ifthenelse{\equal{\draft}{1}}{{\color{red}{---CJ:#1---}}}{}}
\newcommand{\CM}[1]{\ifthenelse{\equal{\draft}{1}}{{\color{blue}{---CJ:#1---}}}{}}
\newcommand{\YY}[1]{\ifthenelse{\equal{\draft}{1}}{{\color{red}{---YY:#1---}}}{#1}}

\newcommand{\commentout}[1]{}

\usepackage{hyperref}
\hypersetup{pdfpagemode=UseNone}

\usepackage{amsfonts}
\usepackage{amssymb}
\usepackage{amsmath}
\usepackage{latexsym}
\usepackage{amsthm}
\usepackage{fullpage}
\usepackage{color}
\usepackage{verbatim}
\usepackage{comment}
\usepackage{enumerate}
\usepackage{tikz}
\usepackage{authblk}

\usetikzlibrary{arrows}

\newtheorem{theorem}{Theorem}[section]
\newtheorem{proposition}[theorem]{Proposition}
\newtheorem{corollary}[theorem]{Corollary}
\newtheorem{definition}[theorem]{Definition}
\newtheorem{lemma}[theorem]{Lemma}

\def\Tr{\textnormal{Tr}}

\def\Par{\textrm{Par}}
\def\MGuess{\textrm{MGuess}}
\def\MagicKey{\textrm{MagicKey}}
\def\MagicQKD{\textrm{MagicQKD}}
\def\Alice\textnormal{Alice}
\def\Bob\textnormal{Bob}
\def\AliceKey{\textit{AliceKey}}
\def\BobKey{\textit{BobKey}}

\def\Eve{\textit{Eve}}
\def\Perm{\textrm{Perm}}
\def\Ham{\textrm{Ham}}
\def\Guess{\textrm{Guess}}

\title{Parallel Device-Independent Quantum Key Distribution}

\author[1]{Rahul Jain}
\author[2]{Carl A. Miller}
\author[3]{Yaoyun Shi}
\affil[1]{Centre for Quantum Technologies, National University of Singapore, and MajuLab, CNRS-UCA-SU-NUS-NTU International Joint Research Unit, Singapore}
\affil[2]{Joint Center for Quantum Information and Computer Science, University of Maryland, and National Institute of Standards and Technology, USA}
\affil[3]{Aliyun Quantum Laboratory, Alibaba, USA}

\date{}

\begin{document}
\maketitle
\thispagestyle{empty}

\begin{abstract}
\noindent
A prominent application of quantum cryptography is the distribution of cryptographic keys that are provably secure.
Recently, such security proofs were extended by Vazirani and Vidick ({\em Physical Review Letters}, 113, 140501, 2014) to the device-independent (DI) scenario, where the users do not need to trust the integrity of the underlying quantum devices. The protocols analyzed by them and by subsequent authors all require a sequential execution of $N$ multiplayer games, where $N$ is the security parameter. In this work, we prove unconditional security of a protocol where  all games are executed in parallel. 
Besides decreasing the number of time-steps necessary for key generation, this result  reduces the security requirements for DI-QKD by allowing arbitrary information leakage
of each user's inputs within his or her lab. 
To the best of our knowledge, this is the first parallel security proof for a fully device-independent QKD protocol.  Our protocol tolerates a constant level of device imprecision and achieves a linear key rate.
\end{abstract}

\clearpage

\section{Introduction}
Key Distribution (KD) is a task where two parties establish a common secret by communicating through a public channel.
It is a necessary step for symmetric key cryptography (i.e., for protocols that require a shared secret) in a setting where a secure communication channel is initially not available. Thus KD is a primitive
foundational to information security.

Information theoretically secure KD is impossible for classical protocols (i.e., protocols that exchange bits).
Thus all classical solutions must necessarily rely on computational assumptions. Widely used protocols,
such as the Diffie-Hellman-Merkle key exchange protocol~\cite{DH} and those making use of digital signatures (e.g., as in the implementation of Secure Sockets Layer) all rely on the computational security of public key cryptography.
The hardness assumptions underlying all known
public key cryptography are mathematically unproven.  The practical security of these solutions are being challenged on the one hand
by the rapidly increasing and widely available high performance computing power, and on the other hand,
by new insights into the design flaws. For example, Adrian {\em et al.}~\cite{Adrian} recently showed how Diffie-Hellman-Merkle could fail in practice. A further threat to all widely used public-key-based KD protocols is that they are not secure against quantum cryptanalysis. With universal quantum computer within sight~\cite{q_Manifesto} and quantum-resilient protocols yet to emerge,  these challenges call for alternative and fundamentally more secure solutions for KD.

Quantum mechanics provides such a solution. The quantum key distribution (QKD) protocol of Bennett and Brassard~\cite{BB84} and its several subsequent variants have been proved to be unconditionally secure (i.e., against a computationally all-powerful adversary)~\cite{Mayers, LC, Shor:2000,Renner:2005,Ben:2005}. Experimental networks implementing QKD have been developed and deployed with increasingly large scales.
With the rapid advances of quantum information technologies, QKD protocols may be widely adopted in the near future.

A major challenge for QKD (and other quantum information tasks) is that quantum information is extremely fragile. How could a user of a QKD protocol be sure that the quantum devices are operating according to the specifications?  
This consideration motivates the field of device-independent (DI) quantum cryptography, pioneered by Ekert \cite{Ekert:1991} and Mayers and Yao~\cite{Mayers:1998}. The goal of DI quantum cryptography is to develop protocols and prove security
in a strictly black-box fashion, with the only trusted assumption being that quantum physics is correct and complete,
and that the users have the ability to restrict information transmission.  The field has seen enormous success
in recent years, including the achievement of fully device-independent and robust security proofs for QKD \cite{Vazirani:QKD:PRL, Miller:2016, Arnon:2016, Dupuis:2016,
Arnon:2018}.

All the known secure DI-QKD protocols are {\em sequential} in the following sense.
Alice and Bob share a two-part quantum device $D = (D_1, D_2)$, each of which is treated as a black box
which accepts classical inputs and returns classical outputs.  Alice creates a random input $X_1$, gives it to her device $D_1$, and receives
an output  $A_1$.  Meanwhile, Bob gives a random input $Y_1$ to his device $D_2$ and receives an output
$B_1$.  This process is repeated sequentially $N$ times to obtain $X_1, \ldots, X_N, Y_1, \ldots, Y_N, A_1 , \ldots, A_N,
B_1, \ldots, B_N$. 
(These data are then used to determine whether a certain Bell inequality has been violated, and if so,
these registers are then postprocessed using information reconciliation and privacy amplification to obtain the final shared key.)
The sequential assumption means specifically that output $A_i$ is recorded before the device gains knowledge of $X_{i+1}$.

The question addressed by the current paper is the following: is the sequential assumption in DI-QKD necessary?  We show that,
in fact, it can be removed: we prove robust DI-QKD in a more general model where there is no time-ordering assumption on the generation of the
outputs $A = (A_1, \ldots, A_N)$ and $B = (B_1, \ldots, B_N)$.  The devices may be treated as black boxes which receive their input sequences $X = (X_1, \ldots, X_N)$
and $Y = (Y_1, \ldots, Y_N)$ all at once and return output sequences $A_1, \ldots, A_N$ and $B_1, \ldots, B_N$ all at once (parallel repetition).  In particular, we do
not require the assumption that $A_i$ is independent of $X_{i+1}$.
The only necessary
assumption is that the inputs $X_1, \ldots, X_N$ are uniformly random conditioned on any information
outside of Alice's lab, and the inputs $Y_1, \ldots, Y_N$ are
uniformly random conditioned on any information outside of Bob's lab. 

Broadening the model for device-independence allows for more flexible implementations of quantum key distribution --- in particular, our
result shows that before quantum key distribution takes place, arbitrary interaction can be allowed between each player's input sequence and his or her device. 
(Indeed, the input sequences can even be re-used from previous experiments, provided that they are completely unknown to the other player
and the adversary when the protocol begins.)  Our model also allows for any of the Bell experiments in the DI-QKD procedure to be performed simultaneously,
which may open the door to faster implementations.

Our work addresses a general theoretical question: what are the minimal assumptions necessary to generate a uniformly random secret between two players?  The main result shows that, if we can assume perfect private randomness and trusted classical computation for each player, then Bell nonlocality itself is enough to generate shared keys of arbitrary length.

\subsection{The protocol and technical statements}

All DI protocols use {\em nonlocal games} as building blocks. For our protocol, we use the Magic Square game.
\begin{definition}
The Magic Square game (MSG) is a two-player game in which the input alphabet for both players
is $\mathcal{X} = \mathcal{Y} = \{ 1, 2, 3 \}$, the output alphabet for the first player is $\mathcal{A} = \{ 000, 011, 101, 110 \}$
(the set of all $3$ bit strings of even parity),
and the output alphabet for the second player is $\mathcal{B} = \{ 001, 010, 100, 111 \}$ (the set of all $3$ bit strings
of odd parity).  The inputs
are chosen according to a uniform distribution, and the game is won if the inputs $x,y$ and the outputs $a,b$
satisfy $a_y = b_x$. 
\end{definition}
The Magic Square game has optimal quantum winning probability $1$ and optimal classical winning probability $8/9$.

For our device model, we assume that Alice and Bob possess an untrusted two-part quantum device $D = (D_1, D_2)$.  The device
$D_1$ receives input from the set $\mathcal{X}^N$, where $N$ is a positive integer, and gives an output
in the set $\mathcal{A}^N$.  The device $D_2$ receives input from the set $\mathcal{Y}^N$ and yields
output in the set $\mathcal{B}^N$.

\begin{figure}[!h]
\fbox{\parbox{6in}{\textbf{Quantum Key Distribution Protocol ($\MagicQKD$)} \\ \\
\textit{Parties:} Alice, Bob, Eve \\ \\
\textit{Parameters:} \\
\begin{tabular}{rl}
$\epsilon$ : & A rational number from the interval $(0, 1/2]$ \\
$N$ : & A positive integer such that $N \epsilon^2$ is an integer \\
\end{tabular} \\ \\
\textit{Equipment:} \\
\begin{tabular}{rl}
$D = (D_1, D_2)$: & A two-part quantum device for playing $N$ copies of the Magic Square \\ &game \\
$E$: & A quantum system (possessed by Eve) which purifies the initial state \\ &of $D$.   \\
$Z$: & A noiseless authenticated public classical channel.
\end{tabular}
\\ \\
\textit{Procedure:} 
\begin{enumerate}
\item Alice chooses $X_1 , \ldots, X_N \in \{ 1, 2, 3 \}$ uniformly at random, gives $(X_1, \ldots, X_N )$
to $D_1$ as input, and records output $(A_1, \ldots, A_N)$.

\item Bob chooses $Y_1, \ldots, Y_N \in \{ 1, 2, 3  \}$ uniformly at random, gives
$(Y_1, \ldots, Y_n )$ to $D_2$ as input, and records output $(B_1, \ldots, B_N)$. 

\item Alice chooses a random permutation $F \colon \{ 1, 2, \ldots, N \} \to \{ 1, 2, \ldots, N \}$ and broadcasts
it to Bob.  The players apply permutation $F$ to $\{ X_i \}, \{ Y_i \}, \{ A_i \}, \{ B_i \}$. 

\item Alice broadcasts $( X_1 , \ldots, X_{ \epsilon N } )$ and
Bob broadcasts $(Y_1, \ldots, Y_{ \epsilon N } )$.

\item For each $j \in \{ 1, 2, \ldots, \epsilon N \}$, Alice records the bit $R_j := (A_j)_{Y_j}$ (that is, the $(Y_j)$th bit of
$A_j$).  For $j \in \{ 1, 2, \ldots, \epsilon N \}$, Bob records the bit $S_j := (B_j)_{X_j}$.

\item Alice broadcasts $( R_1 , \ldots, R_{ \epsilon^2  N } )$ and
Bob broadcasts $(S_1, \ldots, S_{ \epsilon^2  N  } )$.

\item If the average score at the Magic Square game
on games $1, \ldots,  \epsilon^2 N $ is below $1 - \epsilon$, the protocol aborts.
Otherwise, the protocol succeeds, and Alice's raw key consists of the sequence $(R_1 , \ldots, R_{ \epsilon N  } )$
and Bob's raw key consists of the sequence $(S_1, \ldots, S_{ \epsilon N } )$.

\end{enumerate}
}}
\caption{A protocol for key distribution.}
\label{qkdfig}
\end{figure}

Our parallel DI-QKD protocol, {\em MagicQKD}, is given in Figure~\ref{qkdfig}.  Alice and Bob are the
parties who wish to share a key, and Eve is an adversary.  It is assumed that the untrusted
devices $(D_1, D_2 )$ are initially in a pure state with Eve's side information $E$ (which is the worst-case
scenario) and that Eve has access to any communications between Alice and Bob during the protocol. 
The security parameter $N$ is the number of instances of Magic Square played. 
The parameter $\epsilon$ is a positive rational number.  In our proof we show that there is some fixed
positive value $\epsilon := \epsilon_0$ (not given explicitly) such that the protocol achieves a positive linear rate of key distribution
as $N$ tends to infinity.

Our security proof is based on the following assumptions for the protocol $\MagicQKD$.
\begin{enumerate}[{Assumption} 1.]
\item The behavior of the devices $D_1, D_2$ and the system $E$ is modeled by quantum physics.
\item Alice and Bob have the ability to generate perfect private randomness at steps $1$, $2$, and $3$.
\item Any information broadcast by Alice is perfectly received by both Bob and Eve, and any information broadcast by Bob is perfectly received
by both Alice and Eve.
\item Aside from broadcasts by the players, no information is transmitted from Alice's laboratory (which contains $D_1, X, A, R$) or from Bob's laboratories (which contains $D_2, Y, B, S$) once the protocol has started.
\end{enumerate}

Let $\AliceKey$ denote the raw key $R_1, \ldots,  R_{\epsilon N}$ possessed by Alice at the end of
the protocol $\MagicQKD$, let $\BobKey$ denote the raw key $S_1, \ldots, S_{\epsilon N}$ possessed
by Bob, and let $\Eve$ denote all information possessed by Eve at the conclusion of the protocol (including
her side information $E$ and any information obtained by eavesdropping).  Let $\Gamma$ denote the final
state of $\MagicQKD$, and let $SUCC$ denote the event
that the protocol succeeds.  Then, the smooth min-entropy
$H^\delta_{min} ( \AliceKey \mid \Eve, SUCC )$ measures the number of uniformly random bits that
can be extracted from $\AliceKey$ in Eve's presence, while the smooth zero-entropy
$H_0^\delta ( \AliceKey \mid \BobKey, SUCC )$ measures the least number of bits that Alice needs
to publicly reveal in order for Bob to perform information reconciliation and
reconstruct $\AliceKey$ (see section~\ref{collisionsec} for details).  Therefore,
to show security for 
a quantum key distribution protocol, it suffices to show that the difference between the former quantity
and the latter quantity is lower bounded by $\Omega ( N )$, for some negligible error term $\delta := \delta ( N )$.

Our main result is the following.

\begin{theorem}
\label{mainthm}
There exists a constant $\epsilon := \epsilon_0 > 0$ and functions $\delta := \delta ( N ) \in 2^{-\Omega ( N ) }$ and 
$R(N) \in \Omega ( N)$ such that the following always holds for protocol $\MagicQKD$: either
\begin{eqnarray}
\mathbf{P} ( SUCC ) < \delta
\end{eqnarray}
or 
\begin{eqnarray}
\label{secstatement}
H_{min}^{\delta} ( \AliceKey \mid \Eve, SUCC ) - H_0^{\delta} ( \AliceKey \mid \BobKey, SUCC ) & \geq &
R ( N ).
\end{eqnarray}
\end{theorem}
The proof of this theorem is given in Subsection~\ref{disec}.
This theorem establishes both robustness and a linear rate for $\MagicQKD$.  (The data $\epsilon_0, \delta, R$ are
are not given explicitly and are left for future work.)

We note that in the protocol we have assumed that all entanglement shared by the devices $(D_1, D_2)$ is shared
before the protocol begins.  Practically this may be difficult, since it may require a quantum memory size that grows with $N$.  A model
which requires less quantum memory is shown in Figure~\ref{receiveent}, where the entanglement is periodically
updated during step $1$ of $\MagicQKD$ from an outside entanglement source.  (The source and its channels are
both untrusted, and the only assumption is that the communication is one-way.)  Fortunately this case is also covered
by our analysis: a device which behaves as in Figure~\ref{receiveent} is equivalent to one in which all transmissions
from the entanglement source are sent in advance, and are stored in the components $D_1$ and $D_2$.
This illustrates the generality of the parallel model.

If we measure time by the number of prepare-and-measure steps executed by the devices, then a speed-up over
sequential DI-QKD occurs in Figure~\ref{receiveent} if the devices are capable of winning multiple rounds of the Magic Square game at a single iteration.
Quantifying how this speed-up affects the key rate (and also how it increases demands on the devices) is a topic for further research.

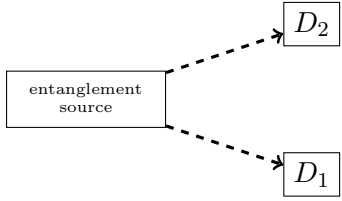
\begin{figure}[!h]
\begin{center}
\begin{tikzpicture}
\node[style=rectangle, draw=black] (ent) at (0,0) {\tiny $\begin{array}{c} \textnormal{entanglement} \\ \textnormal{source} \end{array}$};
\node[style=rectangle, draw=black] (d1) at (3,-1) {$D_1$};
\node[style=rectangle, draw=black] (d2) at (3,1) {$D_2$};
\draw[dashed, very thick, ->] (ent) to (d1);
\draw[dashed, very thick, ->] (ent) to (d2);
\end{tikzpicture}
\end{center}
\caption{A device model in which Alice's and Bob's device receive entanglement from an external source.  The dashed arrows indicate
public one-way communication.}
\label{receiveent}
\end{figure}

\subsection{Security analysis and proof techniques}
\label{secsubsec}

In order to achieve secure parallel DI-QKD, there are two challenges that must be met
simultaneously.  The first is that the parallel scenario opens up the possibility
of correlated cheating strategies by the devices (with correlations going both
``forward'' and ``backward'' between rounds) and one must show a linear amount
of entropy in the key bits despite such strategies.  The second is that the linear rate of
entropy in the raw key must not only be positive; it must be larger than the amount
of entropy that is lost during information reconciliation.

To meet these challenges we made two specific choices in $\MagicQKD$, which differentiate
our protocol from protocols for sequential DI-QKD.  The first is that we use the Magic Square
game, which has special properties for parallel DI-QKD (discussed below).  The second is that the raw keys are only computed
from a randomly chosen subset of the rounds.  This allows us to decrease the amount 
of information that is revealed to Eve during the protocol, and is a necessary assumption for our security proof.

The central challenge when moving from the sequential setting to the parallel
setting is the possibility of new correlations in the behavior of $D_1$ and $D_2$
on separate games.  These correlations can have counter-intuitive properties:
for example, Fortnow gave an example of a two-player game $G$ such that 
$w_c (G^2 ) > w_c ( G )^2$, where $w_c$ denotes the optimal score for classical players
and $G^2$ denotes the game $G$ repeated twice in parallel (see Appendix
A in \cite{Holenstein:2007}).  The same could not
be true in the sequential setting -- the optimal score for $G$ repeated twice in sequence
must be exactly $w_c ( G )^2$.  Thus the parallel assumption opens up new demands for
cheating and requires new techniques.

A technique that has been highly successful for the parallel repetition problem
is based on bounding the amount of information
that players learn about one another's inputs when we condition on events that
depend on a limited number of other games \cite{Raz:1998}. This technique was brought
into the quantum context in \cite{Chailloux:2014, Jain:2014, Chung:2015, Bavarian:2015arXiv},
and allows the proofs of exponentially vanishing upper bounds for the quantum winning probability of $G^N$ for certain broad classes 
of games.  A useful consequence of this technique, which is implicit in \cite{Chailloux:2014, Jain:2014, Chung:2015, Bavarian:2015arXiv}, is that for some games $G$ the behavior of parallel
players on a randomly chosen subset of rounds cannot be much better than the behavior of sequential players.

We apply this technique for parallel repetition to prove security for $\MagicQKD$.  
Specifically, we show that the collision entropy
$H_2 ( \AliceKey\mid \Eve )$ (which, as a well known fact, provides a lower bound on $H_{min}^\epsilon ( \AliceKey \mid \Eve )$) can be expressed in terms of  the winning probability of the ``doubled'' version of the Magic Square game (\MGuess) shown in Figure~\ref{magic6}.
In this expanded game, players Alice, Bob, Alice$'$, and Bob$'$ try to win the Magic Square game
while also trying to guess one another's inputs and outputs.  By the techniques of
\cite{Chailloux:2014, Jain:2014, Chung:2015, Bavarian:2015arXiv}, the probability of winning this doubled game
on $\epsilon N$ randomly chosen rounds in an $N$-fold parallel repetition is not much more than
the probability of winning $\epsilon N$ instances of the games independently.  This fact is the basis for our security claim.

We  also make use of a technique from sequential device-independent
quantum cryptography \cite{Miller:2016, Dupuis:2016}: each time players who are generating random numbers fail
at a single instance of a game, we introduce additional artificial randomness to
compensate for the failure (here represented
by the register $T$ in Figure~\ref{symmagic}).  This artificial randomness (which is useful
for induction) is used only for intermediate steps in the proof and is removed before the final security claim.
This aspect of the proof is important for proving noise tolerance in $\MagicQKD$.

We note that our proof makes use of all of the following properties
of the Magic Square game: (1) it is perfectly winnable by a quantum strategy, (2) its input distribution is uniform, and
 (3) an optimal strategy yields perfectly correlated random bits between Alice and Bob.  (As a consequence of (3),
 there is a positive rate of min-entropy in the raw key bits in $\MagicQKD$, while the communication cost for information reconciliation tends to $0$ when the noise tolerance is lowered, thus enabling a positive key rate.)  The Magic Square game is the simplest game
 that we know of which satisfies all of these properties.  A natural next step
is to study which other games can be used for parallel DI-QKD.
 
After our result was publicized, Thomas Vidick~\cite{Vidick:2017} sketched an alternate proof of DI-QKD,
using a strengthened parallel repetition result that appeared after our result \cite{Bavarian:2015}.  
Vidick's approach uses the class of ``anchored'' games introduced in 2015 \cite{Bavarian:2015arXiv}.
With this approach one can replace  Alice$'$ and Bob$'$ in $\MGuess$ with a single party, and a lower bound on $H_{min}$ (rather than $H_2$) follows via parallel repetition. The protocol in \cite{Vidick:2017} is a version of our protocol which retains the crucial features discussed above.  A comparison between
the rates achieved by these two approaches is a topic for further research.

\medskip

\paragraph{Organization.} Section~\ref{notationsec} establishes notation for our proofs.  Section~\ref{collisionsec} provides
the basis for our interpretation of collision entropy as the winning probability of a ``doubled'' game.
Section~\ref{guesssec} defines the doubled Magic Square game and proves an upper bound
on its winning probability.  Section~\ref{keysec} gives the proof of the central 
security claims.  The appendix proves supporting propositions, including the parallel
repetition result derived from \cite{Chailloux:2014, Jain:2014, Chung:2015, Bavarian:2015arXiv}.

\section{Notation and Preliminaries}

\label{notationsec}

Some of the notation in this section is based on \cite{Tomamichel:2016}. 
If $T$ is a finite set, let $\Perm ( T )$ denote the set of permutations of $T$.   
If $t \in T$, then we write 
can $T \smallsetminus \{ t \}$ to denote the complement of $t$, or if the set $T$ is understood from
the context, we simply write $\widehat{t}$ for $T \smallsetminus \{ t \}$. 

Let $\mathcal{D} ( T )$ denote the set of probability distributions
on the finite set $T$, and let $\mathcal{S} ( T )$ denote the set of subnormalized probability distributions.
If $p, q \in \mathcal{S} ( T )$ are subnormalized distributions let 
\begin{eqnarray}
\Delta ( p, q ) & = & \frac{1}{2} \left( \sum_{t \in T} \left| p ( t ) - q ( t ) \right| + \left| x - y \right| \right)
\end{eqnarray}
where $x := \sum_{t \in T} p ( t)$ and $y := \sum_{t \in T} q ( t )$ respectively.
The function $\Delta$ is a metric on $\mathcal{S} ( T )$.

If $x_1, \ldots, x_N$ and
$y_1, \ldots, y_N$ are binary sequences, let $\Ham ( \mathbf{x} , \mathbf{y} )$ denote the Hamming
distance between $\mathbf{x}$ and $\mathbf{y}$.  The following lemma will be useful in a later proof.
For any $t \in [0, 1]$, let $H ( t )$ denote the Shannon entropy quantity: $H(t)  = - t \log t - (1-t) \log (1-t)$.

\begin{proposition}
\label{hammingprop}
For any $c \in  [0, 1/2]$ and any positive integer $N$, let $L_{c, N}$ denote
the number of $N$-length binary strings whose sum is less than or equal to $cN$.  Then,
$L_{c,N} \leq 2^{N H ( c ) }$.
\end{proposition}

\begin{proof}
We have $L_{c, N} = \sum_{0 \leq i \leq cN} {{N}\choose{i}}$.  The desired inequality
is given in Theorem 1.4.5 in \cite{van:1998}.
\end{proof}

\subsection{Quantum states and operations}

A \textit{quantum register} (or simply \textit{register}) is a finite-dimensional complex Hilbert space with
a fixed orthonormal basis.
We use Roman letters (e.g., $B$) to denote quantum registers.  Given two quantum registers $Q, Q'$,
we will sometimes write $Q Q'$ for the tensor product $Q \otimes Q'$.

If $\mathcal{S}$ is a finite set, an \textit{$\mathcal{S}$-valued quantum register}
is quantum register that has a fixed isomorphism with $\mathbb{C}^\mathcal{S}$.
If $Q$ is a quantum register, let $\mathcal{L} ( Q ), \mathcal{H} ( Q ) , \mathcal{P} ( Q ),
\mathcal{S} ( Q)$, and $\mathcal{D} (  Q  )$, denote, respectively, the sets of
linear, Hermitian, positive semidefinite, subnormalized positive semidefinite (trace $\leq 1$)
and normalized positive semidefinite operators on $Q$.
A \textit{state}
of $Q$ is an element of $\mathcal{D} ( Q )$.  Elements of
$\mathcal{S} ( Q )$ are referred to as \textit{subnormalized states} of $Q$.
A \textit{reflection}
is a Hermitian operator whose eigenvalues are contained in $\{ -1 , 1 \}$.

For any quantum register $Q$, the symbol $I^Q$ denotes the identity operator on $I$, and
$U^Q$ denotes the completely mixed state $I^Q / (\textnormal{dim} ( Q ) )$.

If $Q, Q'$ are quantum registers, the set $\mathcal{L} ( Q )$ has a natural embedding into $\mathcal{L} ( Q \otimes Q')$
by tensoring with $I_{Q'}$. 
We use this embedding implicitly: if $T \in \mathcal{L} ( Q )$ and $\Phi \in \mathcal{D} ( Q \otimes Q' )$,
then $T ( \Phi )$ denotes $(T \otimes I_{Q'}) \Phi$.

Note that if $Q$ is a $\mathcal{Q}$-valued register and $R$ is an $\mathcal{R}$-valued register, then
any function $f \colon \mathcal{Q} \to \mathcal{R}$ determines a process
from $Q$ to $R$ via
\begin{eqnarray}
Z \mapsto \sum_{r \in \mathcal{R} } \left| r \right> \left< r \right| \cdot \left< q \right| Z \left| q \right>.
\end{eqnarray}
We may denote this process by the same letter, $f$.

A \textit{copy} of a register $Q$ is a register $Q'$ with the same dimension with a fixed
isomorphism $Q' \cong Q$.  If $\Gamma \in \mathcal{P} ( \Gamma )$ is a state, then the \textit{canonical purification}
of $\Gamma$ is the projector $\Phi$ on $Q \otimes Q'$ onto the one-dimensional space spanned by $\left( \sqrt{\Gamma} \otimes I_{Q'} \right) \left( \sum_i e_i \otimes e_i \right) \in Q \otimes Q'$,
where the sum is taken over all standard basis elements $e_i$.  We then have $\Phi^Q = \Gamma$ and
$\Phi^{Q'} = \Gamma^\top = \overline{\Gamma}$ under the fixed isomorphism $Q \cong Q'$.

A \textit{measurement} on a register $Q$ is an indexed set $\{ M_i \}_{i \in \mathcal{I}} \subseteq \mathcal{P} ( Q )$ which sums to the
identity.  A \textit{measurement strategy} on $Q$ is a collection of measurements on $Q$ that all have the same index set.

We will use lower case Greek letters (e.g., $\gamma$) to denote complex vectors,
and either uppercase Greek letters (e.g., $\Gamma$) or Roman letters to denote Hermitian operators on finite-dimensional Hilbert spaces.
If $\Gamma$ is a Hermitian operator on a tensor product space $W \otimes V$, then 
$\Gamma^V$ denotes the operator
\begin{eqnarray}
\Gamma^V & := & \Tr_W \Gamma.
\end{eqnarray}
Alternatively we may write $\Gamma^{\widehat{W}}$ for $\Tr_W \Gamma$.
If $T$ is a projector on $W$, let
\begin{eqnarray}
\Gamma_T & = & ( I \otimes T ) \Gamma ( I \otimes T )
\end{eqnarray}
and if $\Tr ( \Gamma_T ) \neq 0$, let $\Gamma_{\mid T} = \Gamma_{T} / \Tr ( \Gamma_T)$.

If $R$ is a register whose values are real numbers, and $\psi$ is a classical state
of $R$, then $\mathbb{E}_\psi [R]$ denotes the expected value of $R$.  
If $\mu$ is a probability distribution on a finite set $\mathbb{Z}$, and $f \colon \mathbb{Z} \to \mathbb{R}$
is a function, then $\mathbb{E}_{z \leftarrow \mu} [f ( z )]$ denotes the expected value of $f ( z )$ under $\mu$.

If $\Phi$ is a positive semidefinite operator, then $\Phi^r$ denotes the operator
that arises from applying the function
\begin{eqnarray}
f ( x ) & = & \left\{ \begin{array}{rl} x^r & \textnormal{ if } x > 0 \\
0 & \textnormal{ if } x = 0. \end{array} \right. 
\end{eqnarray}
to the eigenvalues of $\Phi$.

We make free use of the following shorthands.  If $x_1, \ldots, x_N$ is a sequence,
then the boldface letter $\mathbf{x}$ denotes $(x_1, \ldots, x_N )$.  If
$X_1, \ldots, X_N$ are quantum registers, then $X$ denotes $X_1 X_2 \cdots X_N$.
We write $X_{i \ldots j}$ for the registers $X_{i} X_{i+1} \cdots X_j$.  
If $\{ Y_i^j \}$ is an array of registers, then $Y_i = \{ Y_i^j \}_j$ and $Y^j = \{ Y_i^j \}_i$.
The expression $X_{\widehat{i}}$ denotes the set $\{ X_k \}_{k \neq i}$.

\subsection{Distance measures}

 If $\Gamma_1, \Gamma_2 \in \mathcal{D} ( Q )$
for some quantum register $Q$, let
\begin{eqnarray}
\Delta ( \Gamma_1, \Gamma_2 ) & = & \frac{1}{2} \left\| \Gamma_1 - \Gamma_2 \right\|_1 \\
F ( \Gamma_1 , \Gamma_2 ) & = & \left\| \sqrt{\Gamma_1} \sqrt{\Gamma_2} \right\|_1 \\
P ( \Gamma_1 , \Gamma_2 ) & = & \sqrt{1 - F ( \Gamma_1 , \Gamma_2 )^2 }.
\end{eqnarray}
For $\Lambda_1, \Lambda_2 \in \mathcal{S} (Q )$, let $\left[ \Lambda_i \right]$ denote the density
operator\footnote{That is, the operator on $Q \oplus \mathbb{C}$ given by $\left[ \begin{array}{cc} \Lambda_i  \\ & 1 - \Tr ( \Lambda_i ) \end{array}
\right]$.} $\Lambda_i \oplus (1 - \Tr ( \Lambda_i) )$.  Let
\begin{eqnarray}
\Delta ( \Lambda_1 , \Lambda_2 ) & = & \Delta ( \left[ \Lambda_1 \right] , \left[ \Lambda_2 \right] ) \\ 
F ( \Lambda_1 , \Lambda_2 ) & = & F ( \left[ \Lambda_1 \right] , \left[ \Lambda_2 \right] ) \\ 
P ( \Lambda_1 , \Lambda_2 ) & = & P ( \left[ \Lambda_1 \right], \left[ \Lambda_2 \right] ) 
\end{eqnarray}
The functions $P$ (purified distance) and $\Delta$ (generalized trace distance) are metrics on $\mathcal{S} ( Q )$, and $\Delta \leq P \leq \sqrt{2 \Delta}$.
If $\Lambda_1$ and $\Lambda_2$ are both pure, then $P = \Delta$.
Both quantities $P$ and $\Delta$ satisfy data processing inequalities. (See Chapter 3
in \cite{Tomamichel:2016}).

\subsection{Games}

An \textit{$n$-player nonlocal game} $G$ with input alphabets $\mathcal{X}^1, \ldots , \mathcal{X}^n$
and output alphabets $\mathcal{A}^1, \ldots, \mathcal{A}^n$ is a probability distribution
\begin{eqnarray}
p \colon \prod_i \mathcal{X}^i \to [0, 1 ]
\end{eqnarray}
together with a predicate
\begin{eqnarray}
L \colon \prod_i \mathcal{X}^i \times \prod_i \mathcal{A}^i \to \{ 0, 1 \}.
\end{eqnarray}
Such a game is \textit{free} if $p$ is a uniform distribution.
Let $G^N$ denote the $N$-fold parallel repetition of $G$ (i.e., the game with input alphabets $(\mathcal{X}^i)^N$, output alphabets
$(\mathcal{A}^i)^N$, probability distribution $p ( \mathbf{x} ) = p ( \mathbf{x}_1) \cdot \ldots \cdot p ( \mathbf{x}_n )$, and
predicate $L ( \mathbf{x}, \mathbf{a} ) = \bigwedge_{i=1}^N L ( \mathbf{x}_i, \mathbf{a}_i )$).

A \textit{measurement strategy} for a game $G$ is a family $\{ \{ M_{\mathbf{a} \mid \mathbf{x}} \}_{\mathbf{a}} \}_{\mathbf{x}}$ of $\mathcal{A}$-valued measurements, indexed
by $\mathcal{X}$, on a quantum register $Q = Q_1 \otimes \cdots \otimes Q_n$, where each measurement
operator $M_{\mathbf{a} \mid \mathbf{x}}$ is given by
\begin{eqnarray}
M_{\mathbf{a} \mid \mathbf{x}} = M^1_{a^1 \mid x^1} \otimes \ldots \otimes M^n_{a^n \mid x^n}
\end{eqnarray}
where $\{ M^i_{a^i \mid x^i} \}_{a^i}$ is  a measurement on $Q^i$.

\begin{figure}[h]
\fbox{\parbox{6in}{\textbf{The Parallel Repetition Process ($\Par$)} \\ \\
\textit{Parameters:} \\
\begin{tabular}{rl}
$N$: & A positive integer \\
$G$: & An $n$-player game with input alphabet $\mathcal{X} = \mathcal{X}^1 \times \ldots \times \mathcal{X}^n$ \\ & and output alphabets
$\mathcal{A} = \mathcal{A}^1 \times \ldots \times \mathcal{A}^n$ \\
$\{ M_{\mathbf{a} \mid \mathbf{x}} \}_{\mathbf{a}, \mathbf{x}}$:
& A measurement strategy for $G^N$ on an $n$-partite register $C^1 \cdots C^n$ \\
$\Phi$: & A state of $C^1 \ldots C^n$.
\end{tabular} \\ \\
\textit{Registers:} \\
\begin{tabular}{rl}
$\{ C^k \mid 1 \leq k \leq n \}$: & Quantum registers (for players $1, 2, \ldots, n$, respectively) \\
$\{ X_j^k \mid 1 \leq j \leq N, 1 \leq k \leq n \}$: & Input registers (where $X_j^k$ is $\mathcal{X}^k$-valued) \\
$\{ A_j^k \mid 1 \leq j \leq N, 1 \leq k \leq n \}$:  & Output registers (where $A_j^k$ is $\mathcal{A}^k$-valued) \\
$\{ W_j \mid 1 \leq j \leq N \}$: & Score registers (bit-valued).
\end{tabular}
\\ \\
\textit{Procedure:} 
\begin{enumerate}
\item Prepare $C$ in state $\Phi$.

\item For each $j \in \{ 1, 2, \ldots, N \}$, choose $( x_j^1 , \ldots, x_j^n ) = ( X_j^1 , \ldots, X_j^n )$ at random according to
the input distribution for game $G$.  

\item For each $i$, apply the measurement $\{ M_{a^i \mid x^i} \}_{a^i}$ to the system $C_i$
and record the result in the registers $(A^i_1, \ldots, A^i_N )$.

\item Let $W_i = 1$ if the $i$th game has been won, and let $W_i = 0$ if the $i$th game has been lost.

\item Choose a permutation $\sigma \in \Perm \{ 1, 2, \ldots, N \}$ uniformly at random.
Apply the permutation $\sigma$ to the registers $\{ X_j \}$, the registers $\{ A_j \}$, and the registers $\{ W_j \}$.
\end{enumerate}
}}
\caption{A process defining the parallel repetition of a game.}
\label{parfig}
\end{figure}
It is helpful to describe a parallel repeated game as a process.  In Figure~\ref{parfig}, we introduce the parallel repetition
process $\Par ( N, G, \mathbf{M}, \Phi )$ associated to a game $G$.  The process $\Par$ includes
a final step which shuffles the different instances of the game according to a randomly chosen permutation.

For any $G$, let $w ( G )$ denote the supremum
quantum score of $G$ (i.e., the supremum of $\mathbf{P} ( W_1 = 1 )$ 
in $\Par ( 1, G , \mathbf{M}, \Phi )$ taken over all initial states $\Phi \in \mathcal{D} ( C )$ and all measurements strategies $\mathbf{M}$).

We will typically refer to states arising from processes as follows: the initial state
will be denoted by $\Gamma^0$, and 
$\Gamma^i$ will refer to the state that occurs after step $i$.  The symbol $\Gamma$ will denote the final state.

The following proposition asserts that if $G$ is a free game, then the winning probability in a small number of
rounds in $\Par$ is not much better than that which could be achieved by sequential players.  This fact is implicit in the
entropy approach to parallel repetition given in \cite{Chailloux:2014, Jain:2014, Chung:2015, Bavarian:2015arXiv}.
Since we are not aware of a statement in the literature in the form that we will need, we have given a proof in 
Appendix~\ref{parsec} (see Theorem~\ref{parallelthm}).

\begin{proposition}
\label{keyrandomprop}
Suppose that $G$ is a free nonlocal game.  Then, the registers $W_1, \ldots, W_N$ at the conclusion of process $\Par$
satisfy
\begin{eqnarray}
\label{keyrandineq}
\mathbf{P} ( W_1 = W_2 = \ldots = W_k =1) & \leq & \left[ w(G ) + O_G ( \sqrt{k/N } ) \right]^k.
\end{eqnarray}
for any $k \in \{ 1, 2, \ldots, N \}$.  \qed
\end{proposition}

For our purposes, it is crucial not only that the bound in (\ref{keyrandineq}) is an exponential function, but also
that its base approaches $w(G)$ as $k/N$ approaches zero.

\section{Entropy quantities}

\label{collisionsec}

\begin{definition}
\label{entropydef1}
Let $QR$ be a bipartite quantum register, and let $\Gamma$ be a subnormalized state of $QR$.  Then,
\begin{eqnarray}
h_{min} ( Q \mid R )_{\Gamma} & = & \min_{\substack{\sigma \in \mathcal{S} ( R ) \\
I_Q \otimes \sigma \geq \Gamma}} \Tr ( \sigma) \\
h_2 ( Q \mid R )_{\Gamma} &= & \Tr [ \Gamma (\Gamma^R)^{-1/2}  \Gamma ( \Gamma^R)^{-1/2} ].
\end{eqnarray}
Let
\begin{eqnarray}
\label{minimizationh}
h_{min}^\delta ( Q \mid R )_{\Gamma} & = & \min_{\Gamma'} h_{min} ( Q \mid R )_{\Gamma'}
\end{eqnarray}
where $\Gamma'$ varies over all subnormalized states of $QR$ that
are within distance $\delta$ from $\Gamma$ under the purified distance metric $P$. 
\end{definition}

Note that we can equivalently let the minimization in (\ref{minimizationh}) be taken only over the
the states of $QR$ that have trace no larger than $\Tr ( \Gamma )$, since if $\Tr ( \Gamma' )$ were larger than 
$\Tr ( \Gamma )$, then the scalar multiple $[ \Tr ( \Gamma ) / \Tr ( \Gamma' ) ] \Gamma'$ would be at least as close
to $\Gamma$ as was the original state $\Gamma'$ (see Lemma \ref{renormlemma}).

\begin{definition}
For any subnormalized state $\Lambda$ of a quantum register $T$, let 
\begin{eqnarray}
h ( T )_\Lambda & = & 2^{\Tr [ \Lambda \log \Lambda ]}.  
\end{eqnarray}
and let
\begin{eqnarray}
h ( Q \mid R ) & = & \frac{ h ( QR ) }{h ( R ) }.
\end{eqnarray}
\end{definition}

Additionally, we define some entropy quantities for probability distributions.

\begin{definition}
If $p$ is a probability distribution on a set $S$, let
\begin{eqnarray}
h ( S )_p & = & \prod_{s \in S} p ( s )^{p ( s )}.
\end{eqnarray}
If $q$ is a subnormalized probability distribution on a set $S \times T$, let
\begin{eqnarray}
h_0 ( S \mid T )_q & = & \left( \max_t \left| \left\{ s \in S \mid q ( s, t ) > 0 \right\} \right| \right)^{-1}.
\end{eqnarray}
Let 
\begin{eqnarray}
\label{cneighborhood}
h^\delta_0 ( S \mid T )_q & = & \max_{q'} h_0 ( S \mid T)_{q'},
\end{eqnarray}
where $q'$ varies over all subnormalized probability distributions on $S \times T$ such 
that $\Delta ( q, q' ) \leq \delta$.
\end{definition}

Similar to the definition of smooth min-entropy, in (\ref{cneighborhood}), we can equivalently assume that the minimization is taken over
distributions that are dominated by $q$ (i.e., $q' \leq q$).  For all the entropy quantities specified so far in this
subsection, we let $H_*^* ( \cdot \mid \cdot ) = - \log h_*^* ( \cdot \mid \cdot )$.  (Thus, for example,
$H_{min}^\delta ( Z \mid Y ) = - \log h^\delta_{min} ( Z \mid Y )$.)

If $\Gamma$ is a classical-quantum
state of a bipartite register $ZQ$, and $B$ is a subset of the range $\mathcal{Z}$ of $Z$, then
$\Gamma_B := \Gamma_{P_B}$, where $P \colon \mathbb{C}^\mathcal{Z} \to \mathbb{C}^\mathcal{Z}$ denotes
the projector onto the subspace spanned by $B$, and let
$\Gamma_{\mid B} = \Gamma_{\mid P_B}$. 
When the state is implicit from the context, we may write
\begin{eqnarray}
H_{min} ( Z \mid Q )_B & \textnormal{ and } & H_{min} ( Z \mid Q, B )
\end{eqnarray}
to denote, respectively, 
\begin{eqnarray}
H_{min} ( Z \mid Q )_{\Gamma_B} & \textnormal{ and } & H_{min} ( Z \mid Q)_{\Gamma_{\mid B}},
\end{eqnarray}
and we can use similar notation for the other conditional entropies defined above.

Some of the applications of these quantities are as follows.  Assume that $Z$ is a classical register.
The quantity $H_{min} ( Z \mid Y )$ (quantum conditional min-entropy) is a measure
of the number of bits that can be extracted from $Z$ in the presence of an adversary who possesses $Y$  (see, e.g.,
\cite{Tomamichel:2011}).
The quantity $H ( Z \mid Y )$ (von Neumann entropy) measures the number of bits that can be extracted in the case in which
multiple copies of the state $ZY$ are available (see Chapter 11 in \cite{Nielsen:2010}). 
The quantity
$H_2 ( Z \mid Y )$ is the conditional collision entropy.  In the case where $Y$ is a trivial register,
the quantity $H_2 ( Z \mid Y )$ is the negative logarithm of the probability that two independent samples of $Z$ will agree.
An interpretation of the case where $Y$ is nontrivial will be explained in the next subsection.  

If $Z, Y$ are classical registers with a joint distribution $q$, then
the quantity $H_0 ( Z \mid Y )$ is a measure of the minimum number of bits needed to reconstruct the state
$Y$ from $Z$.  This can be understood as follows: let $M > H_0 ( Z \mid Y )$, and let $R = \{ r \colon \mathcal{Z} \to (\mathbb{Z}_2)^M \}$
be a $2$-universal hash function family.\footnote{That is, $R$ is a family of functions such that for
any distinct $y_1, y_2 \in Y$, the probability that $r ( y_1) = r ( y_2 )$ is no more than $2^{-M}$ when $r$ is chosen
uniformly at random from $R$.}  Suppose that Alice possesses $Z = z$ and Bob possesses $Y = y$, and Alice chooses
$r \in R$ uniformly at random and reveals $r$ and $r ( z )$ to Bob.  Then, except with probability at most $2^{M - H_0 ( Z \mid Y )}$,
there will be only one value in the set $\left\{ z \mid q ( z, y ) > 0 \right\}$ which maps to $r ( z )$ under $r$, and thus
Bob can uniquely determine $z$.

Collision entropy and min-entropy are related by the following 
proposition (see subsection 6.4.1 in \cite{Tomamichel:2016}):
\begin{proposition}
\label{collminprop}
For any quantum registers $RS$, any normalized classical-quantum state $\Gamma$ of $RS$, and any $\delta > 0$,
\begin{eqnarray}
H_{min}^\delta ( R \mid S )_\Gamma \geq H_2 ( R \mid S )_\Gamma - \log ( 2 / \delta^2 ). \qed
\end{eqnarray}
\end{proposition}

\subsection{An operational interpretation of collision entropy for measurements on a pure entangled state}

If $\Gamma$ is a classical-quantum state of a register $ZY$, then a common way 
to describe $h_2 ( Z \mid Y)_{\Gamma}$ is that it is the likelihood that an adversary
who possesses $Y$ can guess $Z$ via the pretty good measurement $\{ (\Gamma^Y )^{-1/2}
\Gamma_{Z=z} ( \Gamma^Y)^{-1/2} \}_z$.  We present an alternative interpretation
which is useful for measuring the randomness obtained from measurements
on an entangled state.  The following proposition refers
to the process $\Guess$ shown in Figure~\ref{guessprocess}.

\begin{figure}[h]
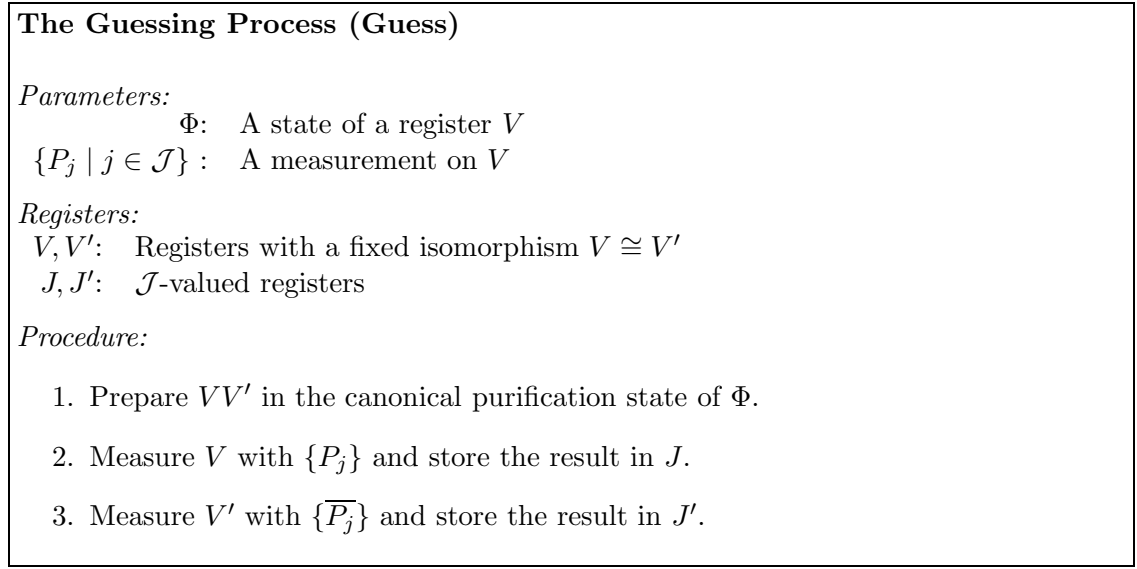

\fbox{\parbox{5.8in}{\textbf{The Guessing Process ($\Guess$)} \\ \\
\textit{Parameters:} \\
\begin{tabular}{rl}
$\Phi$: & A  state of a register $V$\\
$\{ P_j \mid j \in \mathcal{J} \}$ : & A measurement on $V$
\end{tabular}
\\ \\
\textit{Registers:} \\
\begin{tabular}{rl}
$V, V'$: & Registers with a fixed isomorphism $V \cong V'$ \\
$J, J'$: & $\mathcal{J}$-valued registers
\end{tabular}
\\
\\
\textit{Procedure:} 
\begin{enumerate}
\item Prepare $V V'$ in the canonical purification state of $\Phi$.

\item Measure $V$ with $\{ P_j \}$ and store the result in $J$.

\item Measure $V'$ with $\{ \overline{P_j} \}$ and store the result in $J'$.
\end{enumerate}
}}
\caption{A process for guessing measurement outcomes via a purification}
\label{guessprocess}
\end{figure}

\begin{proposition}
\label{collisioninter}
Let $\Gamma^1, \Gamma^2, \Gamma^3$ denote the states that occur after steps 1, 2, and 3, respectively, in the process $\Guess ( \Phi, 
\{ P_j \}_j )$.  Then,
\begin{eqnarray}
\mathbf{P}_{\Gamma^3} ( J = J' ) & = & h_2 ( J \mid V' )_{\Gamma^2}.
\end{eqnarray}
\end{proposition}

\begin{proof}
The states $\left( \Gamma^2 \right)^{J V'}$ and $\left( \Gamma^3 \right)^{J J'}$ are given by
\begin{eqnarray}
\left( \Gamma^2 \right)^{J V'} & = & \sum_j \left| j \right> \left< j \right| \otimes \overline{ \sqrt{\Phi} P_j \sqrt{\Phi}} \\
\left( \Gamma^3 \right)^{J J'} & = & \sum_{j,j'} \left| j j' \right> \left< j j' \right| \Tr [  \overline{ \sqrt{\Phi} P_j \sqrt{\Phi} P_{j'}} ]
\end{eqnarray}
and thus
\begin{eqnarray}
h_2 ( J \mid V')_{\Gamma^2} 
& = & \sum_j \Tr [ \overline{ \sqrt{\Phi} P_j \sqrt{\Phi} \Phi^{-1/2}  \sqrt{\Phi} P_j \sqrt{\Phi} \Phi^{-1/2} } ]  \\
& = & \sum_j \Tr [ \overline{ \sqrt{\Phi} P_j \sqrt{\Phi}  P_j } ] \\
& = & \mathbf{P}_{\Gamma^3} [ J = J' ],
\end{eqnarray}
as desired.
\end{proof}


\section{The Magic Square Guessing Game}
\label{guesssec}

In this section we consider the $6$-player game described in Figure~\ref{magic6}.  In this game,
two pairs of players (Alice and Bob, and Alice$'$ and Bob$'$) each play the Magic Square
game and their inputs and outputs are compared.  There are also additional
players Charlie and Charlie$'$ who receive a random bit and always produce the
same output letter.  (Note that none of the outputs Bob$'$, Charlie, and Charlie$'$ are
used in the scoring rule --- these players are present merely because their inputs are used
in the scoring rule.)

\begin{figure}[h]
\fbox{\parbox{5.6in}{\textbf{The Magic Square Guessing Game ($\MGuess$)} \\ \\
\textit{Players}: Alice, Bob, Charlie, Alice$'$, Bob$'$, Charlie$'$ \\ \\
\textit{Alphabets:} \\
\begin{tabular}{rl}
$\mathcal{X} = \{ 1, 2, 3 \}$: & The input alphabet for Alice and Alice$'$ \\
$\mathcal{Y} = \{ 1, 2, 3 \}$: & The input alphabet for Bob and Bob$'$ \\
$\mathcal{Z} = \{ 0, 1 \}$: & The input alphabet for Charlie and Charlie$'$ \\
$\mathcal{A} = \{ g_1 g_2 g_3 \in \{ 0, 1 \}^3 \mid \bigoplus g_i = 0 \}$: & The output alphabet for Alice and Alice$'$ \\
$\mathcal{B} = \{ g_1 g_2 g_3 \in \{ 0, 1 \}^3 \mid \bigoplus g_i = 1 \}$: & The output alphabet for Bob and Bob$'$ \\
$\mathcal{C} = \{ 0 \}$: & The output alphabet for Charlie and Charlie$'$ \\ \\
\end{tabular} \\
\textit{Probability distribution:} \\
\[p ( x, y, z, x', y', z' ) = \frac{1}{18^2} \hskip0.5in \textnormal{(uniform)}
\]
\\ \\
\textit{Predicate:} \\
\begin{eqnarray*}
L & = & [(x, y, a_y) = (x', y', a'_{y'})] \wedge [(a_y = b_x ) \vee (z = z') ]. 
\end{eqnarray*} \\
The game is won if all three of the following conditions hold:
\begin{enumerate}
\item Alice and Bob's inputs match those of Alice$'$ and Bob$'$.
\item Alice's key bit matches that of Alice$'$.
\item Either $z = z'$ or Alice and Bob win the Magic Square game.
\end{enumerate}
}}
\caption{A game with $6$ players.}
\label{magic6}
\end{figure}

In the game, Alice and Bob are attempting to win the Magic Square game, while Alice$'$ and
Bob$'$ are simultaneously attempting to guess Alice's input, Bob's input, and Alice's key bit.
However, a failure by Alice and Bob at winning the Magic Square game is forgiven if it happens that Charlie and Charlie$'$
have the same output.  (This last rule has the effect of making the game easier
to win.
It underlies the robustness property of our security proof for $\MagicQKD$.)

It is obvious that $w ( \MGuess ) \leq 1/9$, since the probability that Alice's and Bob's inputs
match those of Alice$'$ and Bob$'$ is $1/9$.  We will
prove that in fact $w ( \MGuess )$ is less than $1/9$ minus a positive constant.  This will
be crucial for establishing a nonzero key rate for $\MagicQKD$.

The proof of the next proposition is given in the appendix.
Roughly speaking, the proposition holds because rigidity for the Magic Square game \cite{Wu:2016} implies that
any near-optimal strategy by Alice and Bob involves Alice and Bob performing approximate Pauli measurements
on two approximate EPR pairs shared between them.  The outcomes of such measurements are not guessable by
an outside party (even with entanglement).
Therefore it is impossible for Alice and Bob to achieve a near-perfect score at the Magic Square
game while at the same time allowing Alice$'$ to guess Alice's outcomes.

\begin{proposition}
\label{magicprop}
Let $\MGuess$ denote the game in Figure~\ref{magic6}.  Then, \begin{eqnarray} \label{magicpropineq} w ( \MGuess ) \leq (1/9)  - 0.00035. \end{eqnarray}
\end{proposition}

\begin{proof}
See Appendix~\ref{sgapp}.
\end{proof}

\section{Security Proof}

\label{keysec}

In the current section we give the proof of Theorem~\ref{mainthm}.  Our approach can be roughly understood as follows: 
our upper bound on the winning probability of $\MGuess$ implies, using parallel repetition, that the collision entropy of Alice's and Bob's inputs $X_{1 \ldots \epsilon N} Y_{1 \ldots \epsilon N}$ together with Alice's key bits $R_{1 \ldots \epsilon N}$ is substantially more than that of Alice's and Bob's inputs alone
(for small $\epsilon$).
It follows that, even when we condition on $X_{1 \ldots \epsilon N} Y_{1 \ldots \epsilon N}$ and all of the adversary's other information,
an amount of entropy that is linear in $N$ remains in $R_{1 \ldots \epsilon N}$ (Proposition~\ref{minboundprop}).  On the other hand,
a classical statistical argument shows that the rate of noise between Alice's key bits $R_{1 \ldots \epsilon N }$ and Bob's key bits
$S_{1 \ldots \epsilon N }$ vanishes as $\epsilon \to 0$ (Proposition~\ref{zeroboundprop}).  Combining these facts allows us to
deduce inequality (\ref{secstatement}).

\subsection{An Intermediate Protocol}

In order to show that Alice's raw key in $\MagicQKD$ is sufficiently random, we begin by analyzing the entropy produced by the related protocol $\MagicKey$
in Figure~\ref{symmagic}.  In $\MagicKey$, we use an idea from our previous
work on randomness expansion \cite{Miller:2016, MSUniversal:2015}: when Alice and Bob fail to win the Magic Square game,
we compensate by artificially introducing randomness.  In \cite{Dupuis:2016}, this artificial
randomness is represented by additional registers that have some prescribed entropy,
and we adopt the same style here (by including the registers $T_1, \ldots, T_N$).  We use these
auxilliary registers to establish a lower bound on collision entropy, and the registers will subsequently be dropped.

We begin with the following proposition, which addresses the amount of collision entropy that is collectively contained in 
Alice's and Bob's inputs, Alice's key register, and the auxiliary registers $T_i$.

\begin{theorem}
\label{bigthm}
Let $\Gamma$ be the final state of $\MagicKey$.  Then,
\begin{eqnarray}
\label{collthmineq}
h_2 ( X_{1 \ldots \epsilon N} Y_{1 \ldots \epsilon N } R_{1 \ldots \epsilon N } T_{1 \ldots \epsilon N } \mid E F )_\Gamma & \leq & 
( w ( \MGuess ) + O (\sqrt{ \epsilon})  )^{\epsilon N }.
\end{eqnarray}
\end{theorem}

Note that in the above statement, we are conditioning not only on the register $E$ but also
on the permutation register $F$.

\begin{proof}
We prove this result via an application of Proposition~\ref{collisioninter}.
Upon an appropriate unitary embedding,
we may also assume $E = C'D'$, where $C', D'$ are copies of $C, D$, and that $\Phi$ is the canonical purification of $\Phi^{CD}$.  Suppose that the process $\Par ( N, \MGuess , \mathbf{M}, \Phi)$ is executed with the measurement
strategy\footnote{Here the tensor product respects the following ordering of the players: Alice, Bob, Charlie, Alice$'$, Bob$'$, Charlie$'$.  Charlie
and Charlie$'$ have trivial output, and we treat them as simply performing a unary measurement on a one-dimensional register.}
\begin{eqnarray}
\mathbf{M} = \{ P_\mathbf{a}^\mathbf{x} \otimes Q_\mathbf{b}^\mathbf{y} \otimes I \otimes \overline{P_\mathbf{a'}^\mathbf{x'}} \otimes \overline{Q_\mathbf{b'}^\mathbf{y'}} \otimes  I \},
\end{eqnarray}
For any $m$-subset $Z$ of $\{ 1, 2, \ldots, N \}$, the probability that $\bigwedge_{i \in Z} W_i = 1$ after step 4
in the process $\Par ( N, \MGuess , \mathbf{M}, \Phi)$ is the same as the value of
\begin{eqnarray}
h_2 ( \left\{ X_i Y_i R_i T_i \mid i \in Z \right\} \mid E )
\end{eqnarray}
after step 6 in $\MagicKey$.  The average of the former quantity over all $(\epsilon N)$-subsets is equal to 
the value of $\mathbf{P} ( W_1  \wedge W_2 \wedge \cdots \wedge W_{\epsilon N} )$ at the conclusion
of $\Par ( N, \MGuess , \mathbf{M}, \Phi)$, while the average of the latter quantity is equal
to the expression on the lefthand side of (\ref{collthmineq}).  The desired
result follows from Theorem~\ref{keyrandomprop}.
\end{proof}

\begin{figure}[h!]
\fbox{\parbox{6in}{\textbf{The Magic Square Key Process ($\MagicKey$)} \\ \\
\textit{Parameters:} \\
\begin{tabular}{rl}
$\epsilon$ : & A rational number from $(0, 1/2]$ \\
$N$ : & A positive integer such that $N \epsilon^2$ is an integer \\
$\{ P_{\mathbf{a} \mid \mathbf{x}} \otimes  Q_{\mathbf{b} \mid \mathbf{y}} \}$: &
A measurement strategy for $\textnormal{MagicSquare}^N$ on a bipartite system $CD$ \\
$\Phi$: & A pure state of a tripartite system $CDE$.
\end{tabular}
\\
\\
\textit{Registers:} \\
\begin{tabular}{rl}
$C, D, E$: & Quantum registers (possessed by Alice, Bob, and Eve, respectively) \\
$X_1, \ldots, X_N$: & Alice's input registers ($\{ 1, 2, 3 \}$-valued) \\
$Y_1, \ldots, Y_N$: & Bob's input registers ($\{1, 2, 3 \}$-valued) \\
$A_1, \ldots, A_N$: & Alice's output registers ($\{ 000, 011, 101, 110 \}$-valued) \\
$B_1, \ldots, B_N$: & Bob's output registers ($\{ 001, 010, 100, 111 \}$-valued) \\
$R_1, \ldots, R_N$: & Alice's key bit register \\
$S_1, \ldots, S_N$: & Bob's key bit register \\
$T_1, \ldots, T_N$: & Auxilliary bit registers \\
$F$: & A $\Perm ( \{ 1, 2, \ldots, N \} )$-valued register
\end{tabular}
\\ \\
\textit{Procedure:} 
\begin{enumerate}
\item Prepare $CD$ in state $\Phi$.

\item Choose $\mathbf{x}$ and $\mathbf{y}$ independently and uniformly at random from $\{ 1, 2, 3 \}^N$,
and set $X := \mathbf{x}$ and $Y := \mathbf{y}$.

\item Measure $C$ with $\{ P_{\mathbf{x}}^\mathbf{a} \}_\mathbf{a}$ and store the result in
register $A$.

\item Measure $D$ with $\{ Q_{\mathbf{y}}^\mathbf{b} \}_\mathbf{b}$ and store the result 
in register $B$.

\item For each $i \in \{ 1, 2, \ldots, N \}$, set $R_i$ to be equal to the $(Y_i)$th bit of $A_i$, and 
set $S_i$ to be equal to the $(X_i)$th bit of $B_i$.

\item For each $i \in \{ 1, 2, \ldots, N \}$, if $R_i \neq S_i$, then set $T_i$ to be a independent coin flip.
Otherwise, set $T_i$ to $0$.

\item Choose a random permutation $\sigma \in \Perm \{ 1, 2, \ldots, N \}$ and apply it
to the registers $\{ X_i \}, \{ A_i \}, \{ Y_i \}, \{ B_i \}, \{ R_i \}, \{ S_i \}, \{ T_i \}$.  Store $\sigma$ in the register $F$.
\end{enumerate}
}}
\caption{A protocol for generating a shared key.}
\label{symmagic}
\end{figure}

Next we deduce an upper bound on smooth min-entropy, focusing just
on the registers $R_{1 \ldots \epsilon N} T_{1 \ldots \epsilon N}$.
For compatibility with later derivations, we will take the error parameter to be
$2 \exp ( - \epsilon^4 N  )$.
\begin{corollary}
The following inequality holds:
\begin{eqnarray}
\label{labelcorollary}
H_{min}^{2 \exp ( -  \epsilon^4 N  )} (  R_{1 \ldots \epsilon N }
T_{1 \ldots \epsilon N } \mid X_{1 \ldots \epsilon N } Y_{1 \ldots \epsilon N } E F )_\Gamma
& \geq & \Omega ( \epsilon ) N.
\end{eqnarray}
\end{corollary}
\begin{proof}
By Proposition~\ref{collminprop}, we have
\begin{eqnarray*}
H_{min}^{2 \exp ( -  \epsilon^4 N  )} ( X_{1 \ldots \epsilon N } Y_{1 \ldots \epsilon N } R_{1 \ldots \epsilon N }
T_{1 \ldots \epsilon N } \mid E F )_\Gamma
& \geq & \epsilon N [ \log \frac{ 1}{w ( \MGuess )} - O ( \sqrt{\epsilon})  ] - 2 (\log e ) \epsilon^4 N
\end{eqnarray*}
By Proposition~\ref{magicprop}, $\log [ 1/ w ( \MGuess ) ] > \log (1/9)$, and this bound can be simplified to
\begin{eqnarray*}
H_{min}^{2 \exp ( -  \epsilon^4 N  )} ( X_{1 \ldots \epsilon N } Y_{1 \ldots \epsilon N } R_{1 \ldots \epsilon N }
T_{1 \ldots \epsilon N } \mid E F )_\Gamma
& \geq & N [ ( \log 9 ) \epsilon + \Omega ( \epsilon )  ] .
\end{eqnarray*}
When we condition on the registers $X_{1 \ldots \epsilon N} Y_{1 \ldots \epsilon}$, whose support has size $9^{\epsilon N} =
2^{N \epsilon \log 9 }$, we obtain the bound (\ref{labelcorollary}).
\end{proof}

In the next subsection, we will address conditioning on the event $SUCC$.  For the time being it 
is helpful to condition on a related event.  For any $\delta > 0$, let $WIN ( \delta )$ denote the event that the bit strings
$R_{1\ldots \epsilon N}$ and $S_{1 \ldots \epsilon N}$ differ in at most $\delta ( \epsilon N)$ places.
(That is,
$WIN ( \delta )$ denotes the event that the proportion of wins among the first $\epsilon N$ rounds
is at least $1 - \delta$.) 
Consider the event $WIN ( 2 \epsilon )$.  We have
\begin{eqnarray}
\label{expcorollary}
H_{min}^{2 \exp ( - \epsilon^4 N  )} (  R_{1 \ldots \epsilon N }
T_{1 \ldots \epsilon N } \mid X_{1 \ldots \epsilon N } Y_{1 \ldots \epsilon N } E F )_{\Gamma_{WIN ( 2 \epsilon ) }}
& \geq & \Omega ( \epsilon ) N.
\end{eqnarray}
We assert that a lower bound  in the same form holds when the registers $T_{1 \ldots \epsilon N}$ are omitted.
\begin{corollary}
\label{launchpoint}
The subnormalized state $\Gamma_{WIN ( 2 \epsilon ) }$ satisfies
\begin{eqnarray}
\label{launchpointineq}
H_{min}^{2 \exp ( - \epsilon^4 N  )} (  R_{1 \ldots \epsilon N }
\mid X_{1 \ldots \epsilon N } Y_{1 \ldots \epsilon N } E F )
& \geq & \Omega ( \epsilon ) N .
\end{eqnarray}
\end{corollary}

\begin{proof}
The distribution of the registers $T_{1\ldots \epsilon N}$ under the subnormalized state $\Gamma_{WIN ( 2 \epsilon ) }$ is supported
only on binary strings of Hamming weight less than $2 \epsilon^2 N$.  Thus, by Proposition~\ref{hammingprop}, these registers
are supported on a set of size less than or equal to $2^{H ( 2 \epsilon ) \epsilon N}$.  Therefore we can drop
the registers $T_{1 \ldots \epsilon N}$ from the lefthand side of (\ref{expcorollary}) and
and deduct $H ( 2 \epsilon ) \epsilon N$ from its righthand side, and the inequality is preserved.  Since the term
$H ( 2 \epsilon ) \epsilon N$ is dominated by $\Omega ( \epsilon ) N$,
it may be ignored and the desired result follows.
\end{proof}

\subsection{Device-Independent Quantum Key Distribution}
\label{disec}

We now turn our attention toward the protocol $\MagicQKD$ (Figure~\ref{qkdfig}).
We will prove that $\MagicQKD$ generates a positive key rate.  Our final statement will use
the registers
\begin{eqnarray}
AliceKey & := & R_{1 \cdots \epsilon N} \\
BobKey & := & S_{1 \cdots \epsilon N} \\
Eve & := & X_{1 \cdots \epsilon N} Y_{1 \cdots \epsilon N} R_{1 \cdots \epsilon^2 N } S_{1 \cdots \epsilon^2 N} EF.
\end{eqnarray}
The registers $Eve$ denote the information possessed by Eve at the conclusion of $\MagicQKD$.

We begin by translating Corollary~\ref{launchpoint} into a statement about the success event
for $\MagicQKD$.  Let $SUCC$ denote the event that $\MagicQKD$ succeeds,
and let $SUCC'$ denote the event
that $\MagicQKD$ succeeds \textit{and} the event $WIN ( 2 \epsilon )$ occurs.
\begin{lemma}
\label{succprimelemma}
The events $SUCC'$ and $SUCC$ satisfy
\begin{eqnarray}
\mathbf{P} ( SUCC \wedge \neg SUCC' ) & \leq & e^{- 2 \epsilon^4 N}.
\end{eqnarray}
\end{lemma}

\begin{proof}
We assume $\mathbf{P} ( \neg WIN ( 2 \epsilon ) ) > 0$ (since otherwise the desired assertion is obvious). 
We have
\begin{eqnarray}
\mathbf{P} ( SUCC \wedge \neg SUCC' ) & = & \mathbf{P} ( SUCC \wedge \neg WIN ( 2 \epsilon ) ) \\
\label{succprimestep}
& = & \mathbf{P} ( \neg WIN ( 2 \epsilon ) ) \cdot  \mathbf{P} ( SUCC \mid \neg WIN ( 2 \epsilon )) 
\end{eqnarray}
We consider the second factor in (\ref{succprimestep}).  Let $W_i$ denote the indicator variable for the event that the $i$th game is won.
After conditioning on $\neg  WIN ( 2 \epsilon )$, the only way that $SUCC$ can  occur is if
the average of the variables $W_1 , \ldots, W_{\epsilon^2 N}$ exceeds that of $W_1 \ldots, W_{\epsilon N}$
by at least $\epsilon$.  
By (\cite{Hoeffding:1963}, Theorem 1 and Section 6), if an $\epsilon^2 N$-subset $S$ is chosen at random from a set of
Boolean values $T$ of size $\epsilon N$, then the probability that the average of $S$
will exceed that of $T$ by more than $\epsilon$ is at most $e^{- 2 \epsilon^2 (\epsilon^2 N )}$.
This yields the desired bound.
\end{proof}

As a consequence of Lemma~\ref{succprimelemma}, we have $\Delta ( \Gamma_{SUCC} , \Gamma_{SUCC'} ) \leq 2 \exp ( - 2 \epsilon^4 N)$, and
therefore $P ( \Gamma_{SUCC} , \Gamma_{SUCC'} ) \leq \sqrt{ 4 \exp ( - 2 \epsilon^4 N)} = 2 \exp ( - \epsilon^4 N )$.  
Since $SUCC' \Longrightarrow WIN ( 2 \epsilon )$, $\Gamma_{SUCC'}$ also satisfies inequality (\ref{launchpointineq})
from Corollary~\ref{launchpoint}.  We therefore have by the triangle inequality that
 the state $\Gamma_{SUCC}$ satisfies
\begin{eqnarray}
\label{incompletesecurity}
H_{min}^{4 \exp ( - \epsilon^4 N  )} (  R_{1 \ldots \epsilon N }
\mid X_{1 \ldots \epsilon N } Y_{1 \ldots \epsilon N } E F )
& \geq & \Omega ( \epsilon ) N .
\end{eqnarray}
Conditioning also on the registers $R_{1 \ldots \epsilon^2 N} S_{1 \ldots \epsilon^2 N}$ decreases the quantity on the lefthand
side of (\ref{incompletesecurity}) by at most $2 \epsilon^2 N \leq o ( \epsilon ) N$, and thus we obtain the following result.
\begin{proposition}
\label{minboundprop}
The state $\Gamma_{SUCC}$ at the conclusion of $\MagicQKD$ satisfies
\begin{eqnarray}
H_{min}^{4 \exp ( - \epsilon^4 N  )} ( \AliceKey \mid \Eve )
& \geq & \Omega ( \epsilon ) N .
\end{eqnarray}
\end{proposition}

Meanwhile, by definition, the registers $\AliceKey$ and $\BobKey$ in the state $\Gamma_{SUCC'}$ differ
in at most $2 \epsilon^2 N$ places, and thus by Proposition~\ref{hammingprop}, we have
\begin{eqnarray}
\label{oepsilonineq}
H_0 ( \AliceKey \mid \BobKey )_{SUCC'} & \leq & N \epsilon H ( 2 \epsilon ) \\
& \leq & N o ( \epsilon ).
\end{eqnarray}
Applying Lemma~\ref{succprimelemma} yields the following.
\begin{proposition}
\label{zeroboundprop}
The state $\Gamma_{SUCC}$ at the conclusion of $\MagicQKD$ satisfies
\begin{eqnarray}
H_0^{2 \exp ( - 2 \epsilon^4 N  )} ( \AliceKey \mid \BobKey )
& \leq & o ( \epsilon ) N .
\end{eqnarray}
\end{proposition}

We can now prove our main result.

\begin{proof}[Proof of Theorem~\ref{mainthm}]
Let
\begin{eqnarray}
\delta & = &  2 e^{-\epsilon^4N/3}.
\end{eqnarray}
If $\mathbf{P} ( SUCC ) \geq \delta $, then, by Propositions~\ref{nentprop1} and \ref{nentprop2} in the appendix,
\begin{eqnarray}
&& H_{min}^{\delta} ( \AliceKey \mid \Eve , SUCC ) - H_0^\delta ( \AliceKey \mid \BobKey, SUCC ) \\
\label{penultimate1}
& \geq & H_{min}^{\delta^3/2} ( \AliceKey \mid \Eve )_{SUCC} - H_0^{\delta^2} ( \AliceKey \mid \BobKey )_{SUCC} - \log (1/\delta) \\
\label{penultimate2}
& \geq & N \Omega ( \epsilon )  - N o ( \epsilon ) - [(\log e ) \epsilon^4N/3 + 1] \\
\label{finalquantity}
& \geq & N \Omega ( \epsilon ),
\end{eqnarray}
where in lines (\ref{penultimate1})--(\ref{penultimate2}), we used the fact that the terms $\delta^3/2$ and $\delta^2$ are at least as large as
the respective error terms in Propositions~\ref{minboundprop} and \ref{zeroboundprop}.
We now simply fix $\epsilon := \epsilon_0 > 0$ to be
sufficiently small that the function denoted by $\Omega ( \epsilon )$ in (\ref{finalquantity}) is positive,
and the proof is complete.
\end{proof}

\section{Acknowledgments}
This research was supported in part by the U.~S.~National Science Foundation (NSF) under Awards 1526928, 1500095, and 1717523, when Y.~S.~was at University of Michigan.
Work by R.~J.~on this research was supported by the Singapore
Ministry of Education and the National Research Foundation, through
the Tier 3 Grant ``Random numbers from quantum processes''
MOE2012-T3-1-009 and NRF RF Award NRF-NRFF2013-13.  This paper is partly a contribution
of the U. S. National Institute of Standards and Technology (NIST), and is not subject to copyright in the United States.

\newpage
\appendix

\appendix

\def\Player{\textit{Player}}

\section{Supporting Proofs for Entropy Measures}

The following two lemmas bound the amount that the
purified distance $P ( \sigma , \lambda )$ can increase under scalar multiplication of
the two states $\sigma, \lambda$.  We address a case where the scalar multiplication
makes the trace of the two states equal, and also a case where scalar multiplication normalizes the larger of the two states.

\begin{lemma}
\label{renormlemma}
Let $Q$ be a quantum register, let $\lambda, \sigma \in \mathcal{S} ( Q )$, and let $r = \Tr ( \lambda), s = \Tr ( \sigma)$.
Suppose that $s \geq r > 0$.  Then,
\begin{eqnarray}
\label{prenorm}
P ( (r/s) \sigma , \lambda ) & \leq & P ( \sigma, \lambda ).
\end{eqnarray}
\end{lemma}

\begin{proof}
Let $\Lambda, \Sigma$ be the normalizations of $\lambda, \sigma$.
Using the Cauchy-Schwartz inequality, we have the following.
\begin{eqnarray}
F ( \sigma, \lambda ) & = & \sqrt{ (1 - r ) (1 - s ) } + \sqrt{ rs} \left\| \sqrt{ \Lambda } \sqrt { \Sigma } \right\|_1  \\
& = & \sqrt{ (1-r) + r \left\| \sqrt{ \Lambda } \sqrt { \Sigma } \right\|_1 } \sqrt{ (1-s) + s \left\| \sqrt{ \Lambda } \sqrt { \Sigma } \right\|_1 } \\
& \leq & \sqrt{ (1-r) + r \left\| \sqrt{ \Lambda } \sqrt { \Sigma } \right\|_1 } \sqrt{ (1-r) + r \left\| \sqrt{ \Lambda } \sqrt { \Sigma } \right\|_1 } \\
& = & F ( (r/s) \sigma, \lambda ),
\end{eqnarray}
Inequality (\ref{prenorm}) follows.
\end{proof}

\begin{lemma}
\label{renormlemma2}
Under the assumptions of Lemma~\ref{renormlemma}, the following inequality also holds.
\begin{eqnarray}
\label{prenorm2}
P ( \sigma/s , \lambda/s ) \leq
\sqrt{(2/s) P ( \sigma, \lambda)}.
\end{eqnarray}
\end{lemma}

\begin{proof}
Note that the quantity
\begin{eqnarray}
\Delta ( c \sigma, c \lambda ) & = & c \left\| \sigma - \lambda \right\|_1 + c \left| \Tr \sigma - \Tr \lambda \right|
\end{eqnarray}
is linear in $c$.  We have
\begin{eqnarray}
P ( \sigma/s , \lambda/s ) \leq \sqrt{ 2 \Delta ( \sigma/s, \lambda/s ) } \leq \sqrt{(2/s) \Delta ( \sigma, \lambda ) }
\leq \sqrt{(2/s) P ( \sigma, \lambda)},
\end{eqnarray}
as desired.
\end{proof}

Now we use Lemma~\ref{renormlemma2} to address how smooth min-entropy behaves under normalization.
\begin{proposition}
\label{nentprop1}
Let $\sigma \in \mathcal{S} ( QR)$ be a nonzero state, let $\Sigma$ be its normalization, and let $\delta > 0$.  Then,
\begin{eqnarray}
\label{deltasquared}
H_{min}^{\delta} ( Q \mid R )_{\Sigma} & \geq & H_{min}^{\delta^2\Tr ( \sigma)/2} ( Q \mid R )_\sigma - \log (1/\Tr ( \sigma)).
\end{eqnarray}
\end{proposition}

\begin{proof}
Let $s = \Tr ( \sigma)$.
Find a state $\sigma'$ satisfying
satisfying $\Tr ( \sigma' ) \leq s$ and $P ( \sigma' , \sigma ) \leq \delta^2 s /2$ such that
\begin{eqnarray}
H_{min} ( Q \mid R )_{\sigma'} = H_{min}^{\delta^2s/2} ( Q \mid R )_\sigma.
\end{eqnarray}
(See the discussion following Definition~\ref{entropydef1}.)
The conditional min-entropy of $\sigma'/s$ is then given by the expression on the righthand
side of (\ref{deltasquared}), and by Lemma~\ref{renormlemma2},
\begin{eqnarray}
P ( \sigma'/s , \sigma/s ) \leq \sqrt{ 2 P ( \sigma', \sigma ) / s } \leq \delta.
\end{eqnarray}
Inequality (\ref{deltasquared}) follows.
\end{proof}

The next proposition similarly addresses how $H_0^\delta$ behaves under normalization.

\begin{proposition}
\label{nentprop2}
Let $q$ be a nonzero subnormalized probability distribution on $S \times T$, where $S, T$ are finite sets, and let 
$s$ be the norm of $q$.  Let $\delta > 0$.  Then,
\begin{eqnarray}
 H_0^{\delta} ( S \mid T )_{q/s} & = & H_0^{s \delta} ( S \mid T )_q.
\end{eqnarray}
\end{proposition}

\begin{proof}
This is immediate from the linearity of the distance function $\Delta$.
\end{proof}

\section{Proof of Proposition~\ref{magicprop}}

\label{sgapp}

Our proof builds off of steps from the proof of rigidity for the Magic Square
game \cite{Wu:2016}.  We will reproduce the fact that
any near-optimal strategy for the Magic Square must involve approximately anti-commuting measurements,
and use that fact  to deduce inequality (\ref{magicpropineq}).

Let $\{ F_x^a \}, \{ G_y^b \}, \{ {F'}_x^a \}$ be the measurements
used by Alice, Bob, and Alice$'$, respectively, which we will assume (without loss of generality) to be projective, and let $\Phi$
denote their shared state, which we will assume to be pure: $\Phi = \phi \phi^*$.  For $i,j \in \{ 1, 2, 3 \}$, let $F_{ij}$ denote the reflection
operator
\begin{eqnarray}
F_{ij} & = & \sum_{\substack{a \in \mathcal{X} \\ a_j = 0} } F_i^a - \sum_{\substack{a \in \mathcal{X} \\ a_j = 1} } F_i^a,
\end{eqnarray}
define $F'_{ij}$ similarly in terms of $\{ {F'}_x^a \}$, and let
\begin{eqnarray}
G_{ij} & = & \sum_{\substack{b \in \mathcal{X} \\ b_i = 0} } G_j^b - \sum_{\substack{b \in \mathcal{Y} \\ b_i = 1} } G_j^b.
\end{eqnarray}
Note that $F_{ij}$ and $F_{ik}$ always commute and
$F_{i1} F_{i2} F_{i3} = I$, that $G_{ij}$ and $G_{kj}$ always commute and $G_{1j} G_{2j} G_{3j} = -I$,
and similar relationships hold for $F'_{ij}$.

Let
\begin{eqnarray}
\delta & = & \mathbf{P} ( A_Y \neq B_X )  \\
\delta_{ij} & = & \mathbf{P} ( A_Y \neq B_X \mid X = i, Y = j) 
\end{eqnarray}
and
\begin{eqnarray}
\epsilon & = & \mathbf{P} ( A_Y \neq A'_{Y'} \mid X = X', Y = Y' ), \\
\epsilon_{ij} & = & \mathbf{P} ( A_Y \neq A'_{Y'} \mid X = X' = i, Y = Y' = j ).
\end{eqnarray}
Note that
\begin{eqnarray}
\mathbf{P} ( L = 1 ) & \leq & \mathbf{P} ( X = X', Y = Y') \mathbf{P} ( Z = Z' \vee A_Y = B_X \mid X = X', Y = Y') \\
& = & (1/9) (1 - \delta/2),
\end{eqnarray}
and also
\begin{eqnarray}
\mathbf{P} ( L = 1 ) & \leq & \mathbf{P} ( X = X', Y = Y') \mathbf{P} ( A_Y = A'_{Y'} \mid X = X', Y = Y' ) \\
& = & (1/9) ( 1 - \epsilon ).
\end{eqnarray}
Thus,
\begin{eqnarray}
\label{ellone}
\mathbf{P} ( L = 1 ) & \leq & (1/9) - (1/9) \max \left\{ \epsilon, \delta/2 \right\},
\end{eqnarray}
and to complete the proof it suffices to find a general lower bound for $\max \left\{ \epsilon, \delta/2 \right\}$.

Note that for any $i,j  \in \{ 1, 2, 3 \}$,
\begin{eqnarray}
\mathbf{P} ( A_Y \neq B_X \mid X = i, Y = j ) & = & (1 - \phi^* F_{ij} \otimes G_{ij} \phi )/2 \\
& = & \left\| \phi - (F_{ij} \otimes G_{ij}) \phi \right\|^2/4,
\end{eqnarray}
and thus
\begin{eqnarray}
\left\| \phi - (F_{ij} \otimes G_{ij}) \phi \right\| & = & 2 \sqrt{\delta_{ij}}.
\end{eqnarray}
By similar reasoning,
\begin{eqnarray}
\left\| \phi - (F_{ij} \otimes F'_{ij}) \phi \right\| & = & 2 \sqrt{\epsilon_{ij}}.
\end{eqnarray}

We exploit the approximate anti-commutativity relations for $\{ F_{ij} \}$ which are proven in \cite{Wu:2016}.  We have the following.
\begin{eqnarray*}
\left\| (F_{11} F_{22}  ) \phi - (F_{11} \otimes G_{22}  ) \phi \right\| & \leq & 2 \sqrt{\delta_{22}} \\
\left\| (F_{11} F_{22}  ) \phi - (G_{22} G_{11}  ) \phi \right\| & \leq & 2 (\sqrt{\delta_{22}} + \sqrt{\delta_{11}})  \\
\left\| (F_{11} F_{22}  ) \phi - (G_{12} G_{32} G_{31}  G_{21}  ) \phi \right\| & \leq & 2 (\sqrt{\delta_{22}} + \sqrt{\delta_{11}}) \\
\left\| (F_{11} F_{22}  ) \phi - (F_{21} F_{31} F_{32}\otimes G_{12} ) \phi \right\| & \leq & 2 (\sqrt{\delta_{22}} + \sqrt{\delta_{11}} +
\sqrt{\delta_{32}} + \sqrt{\delta_{31}} + \sqrt{\delta_{21}}) \\
\left\| (F_{11} F_{22}  ) \phi - (F_{21} F_{33}  \otimes G_{12}  ) \phi \right\| & \leq & 2 (\sqrt{\delta_{22}} + \sqrt{\delta_{11}} +
\sqrt{\delta_{32}} + \sqrt{\delta_{31}} + \sqrt{\delta_{21}}) \\
\left\| (F_{11} F_{22}  ) \phi - (F_{21}  \otimes G_{12} G_{33} ) \phi \right\| & \leq & 2 (\sqrt{\delta_{22}} + \sqrt{\delta_{11}} +
\sqrt{\delta_{32}} + \sqrt{\delta_{31}} + \sqrt{\delta_{21}} + \sqrt{\delta_{33}}) \\
\left\| (F_{11} F_{22}  ) \phi - (- F_{21}  \otimes G_{12} G_{13} G_{23}  ) \phi \right\| & \leq & 2 (\sqrt{\delta_{22}} + \sqrt{\delta_{11}} +
\sqrt{\delta_{32}} + \sqrt{\delta_{31}} + \sqrt{\delta_{21}} + \sqrt{\delta_{33}}) \\
\left\| (F_{11} F_{22}  ) \phi - (- F_{21}  F_{23} F_{13} F_{12}  ) \phi \right\| & \leq & 2 \sum_{ij} \sqrt{\delta_{ij}} \\
\left\| (F_{11} F_{22}  ) \phi - (- F_{22} F_{11} ) \phi \right\| & \leq & 2 \sum_{ij} \sqrt{\delta_{ij}} \\
\left\| (F_{11} F_{22}  ) \phi + ( F_{22} F_{11}  ) \phi \right\| & \leq & 2 \sum_{ij} \sqrt{\delta_{ij}}.
\end{eqnarray*} 
By the concavity of the square root function, this yields
\begin{eqnarray}
\nonumber
\left\| (F_{11} F_{22}  ) \phi + ( F_{22} F_{11}  ) \phi \right\| & \leq & 18 \sum_{ij} \sqrt{\delta_{ij}}/9. \\
\nonumber
& \leq & 18 \sqrt{ \sum_{ij} \delta_{ij} / 9} \\
\label{deltabound} & = & 18 \sqrt{ \delta}.
\end{eqnarray}

We also have the following, in which we make use of the approximate compatibility of the measurements $\{ F_{ij} \}$
and the measurements $\{ F'_{ij} \}$. 
\begin{eqnarray}
\left\| (F_{11} F_{22} \otimes I ) \phi - (F_{11} \otimes F'_{22} \otimes I ) \phi \right\| & \leq & 2 \sqrt{\epsilon_{11}}  \\
\left\| (F_{11} F_{22} \otimes I ) \phi - (G_{11} \otimes F'_{22} \otimes I ) \phi \right\| & \leq & 2 \sqrt{\epsilon_{11}} + 2 \sqrt{\delta_{22}} \\
\left\| (F_{11} F_{22} \otimes I ) \phi - (G_{11} \otimes F_{22} \otimes I ) \phi \right\| & \leq &  4 \sqrt{\epsilon_{11}} + 2 \sqrt{\delta_{22}} \\
\label{almostcommute}
\left\| (F_{11} F_{22} \otimes I ) \phi - (F_{22} F_{11}  \otimes I ) \phi \right\| & \leq & 4 \sqrt{\epsilon_{11}} + 4 \sqrt{\delta_{22}},
\label{secondfineq}
\end{eqnarray}
Combining (\ref{almostcommute}) via the triangle inequality with (\ref{deltabound}) (and using the fact that $(F_{22} F_{11} \otimes I ) \phi$ is a unit vector) yields
\begin{eqnarray}
2 & \leq & 18 \sqrt{\delta} + 4 \sqrt{\epsilon_{11}} + 4 \sqrt{\delta_{22}} 
\end{eqnarray}
By symmetry, we likewise have the following for any $i, j, i', j' \in \{ 1, 2, 3 \}$ with $i \neq i', j \neq j'$:
\begin{eqnarray}
2 & \leq & 18 \sqrt{\delta} + 4 \sqrt{\epsilon_{ij}} + 4 \sqrt{\delta_{i'j'}} 
\end{eqnarray}
Averaging all such inequalities and exploiting the concavity of the square root function, we obtain
\begin{eqnarray}
2 & \leq & 18 \sqrt{\delta} + 4 \sqrt{\epsilon} + 4 \sqrt{\delta},
\end{eqnarray}
which implies
\begin{eqnarray}
\label{oneineq}
1 & \leq & 11 \sqrt{\delta} + 2 \sqrt{\epsilon}.
\end{eqnarray}

From (\ref{oneineq}), we can compute a lower bound on $\max \{ \epsilon, \delta/2 \}$.  If $\epsilon \leq \delta/2$, then, 
\begin{eqnarray}
1 & \leq & 11 \sqrt{\delta} + \sqrt{2 \delta}
\end{eqnarray}
which yields
\begin{eqnarray}
\delta/2 & \geq & (1/2) \cdot (11 + \sqrt{2})^{-2},
\end{eqnarray}
while if $\epsilon \geq \delta/2$, similar reasoning yields
\begin{eqnarray}
\epsilon & \geq & (1/2) \cdot (11 + \sqrt{2})^{-2}.
\end{eqnarray}
Therefore,
\begin{eqnarray}
\max \{ \epsilon, \delta/2 \} & \geq & (1/2) \cdot (11 + \sqrt{2})^{-2}.
\end{eqnarray}
Substituting this value into (\ref{ellone}), we find
\begin{eqnarray}
\mathbf{P} ( L = 1 ) & \leq & (1/9) - (1/18) \cdot (11 + \sqrt{2})^{-2} \\
& \leq & (1/9) - 0.00035,
\end{eqnarray}
as desired.

\section{Randomly chosen rounds in parallel repetition of a free game}

\label{parsec}

In this appendix, we prove that in a parallel repetition of a free game,
the performance of the players on a small number of randomly chosen rounds is not much better
than their performance would have been in a sequential scenario.  
Our proof is a rearrangement of elements from \cite{Chailloux:2014, Jain:2014, Chung:2015, Bavarian:2015arXiv}.

For any
state $\rho$ of a bipartite system $QR$, the mutual information between $Q$ and $R$ and is given by
\begin{eqnarray}
I ( Q \colon R )_\rho & = & H ( QR ) - H ( Q) - H ( R ).
\end{eqnarray}
Let $S ( \rho \| \sigma ) = \Tr [ \rho \log \rho] - \Tr [ \rho \log \sigma ]$ denote the relative
entropy function.  The following relationship holds:
\begin{eqnarray}
\label{altiexp}
I ( Q \colon R )_\rho & = & S ( \rho \| \rho^A \otimes \rho^B ) \\
\end{eqnarray}
Also, the relative entropy function is related to the purified distance as follows: if $\alpha, \beta$ are density operators, then
\begin{eqnarray}
P ( \alpha, \beta ) & \leq & \sqrt{ S ( \alpha \| \beta ) }. 
\end{eqnarray}
(This follows from, e.g., Lemma~5 in \cite{Jain:2003}.) 

Throughout this section, we assume that a free game $G = ( \mathcal{X}, \mathcal{A}, p, L)$, with $w ( G ) > 0$, has been fixed.  (Thus we avoid any need to note the influence of $G$ on error terms.)

\subsection{Preliminaries}

Our first result
asserts (roughly) that if a state $\gamma$ of a bipartite system $TQ$ is dominated
by a small scalar multiple of a state that is uniform on $T$, then $H ( T \mid Q )_\gamma$ must
be close to $\log \left| T \right|$.
\begin{lemma}
\label{subuniformlemma}
Let $\gamma$ be a classical-quantum state of a bipartite system $TQ$ such that
\begin{eqnarray}
\gamma \leq \lambda (U_T \otimes \gamma^Q),
\end{eqnarray}
where $\lambda$ denotes a real number.  Then,
\begin{eqnarray}
H ( T \mid Q )_\gamma & \geq & \log \left| T \right| - 2 \log (1/\lambda)
\end{eqnarray}
\end{lemma}

\begin{proof}
We have $H ( T \mid Q )_\gamma  =  H ( T )_\gamma - I ( T : Q )_\gamma$.
It is obvious that the quantity $H ( T )_\gamma$ is at least $\log \left| T \right| - \log (1/\lambda)$
since the eigenvalues of $\gamma^T$ do not exceed $\lambda/\left| T \right|$.   Thus we need
only prove that $I ( T: Q )_\gamma \leq \log ( 1/\lambda)$.

We can write
\begin{eqnarray}
I ( T : Q ) & = & S ( \gamma \| \gamma^T \otimes \gamma^Q ).  
\end{eqnarray}
Note that the quantity
\begin{eqnarray}
S ( \gamma \|  U_T \otimes \gamma^Q ) - S ( \gamma \| \gamma^T \otimes \gamma^Q )
\end{eqnarray}
is equal to $S ( U_T \| \gamma^T )$, which is nonnegative, and therefore
\begin{eqnarray}
I ( T : Q ) & \leq & S ( \gamma \| U_T \otimes \gamma^Q ).  
\end{eqnarray}
We therefore have the following, using the fact that the logarithm function is operator monotone:
\begin{eqnarray}
I ( T : Q ) & \leq & S ( \gamma \| U_T \otimes \gamma^Q ) \\
& = & \Tr [ \gamma \log \gamma ] - \Tr [ \gamma \log (   U_T \otimes \gamma^Q ) ] \\
& \leq & \Tr [ \gamma \log \gamma ] - \Tr [ \gamma \log ( \gamma / \lambda) ] \\
& = & \log ( 1/\lambda),
\end{eqnarray}
as desired.
\end{proof}

By definition, if two pure bipartite states $\psi, \phi \in \mathcal{D} ( Q \otimes R )$ are
such that $P ( \psi^Q, \phi^Q ) = \delta$, then there is a unitary automorphism of $R$
which maps $\phi$ to a state that is within $\Delta$-distance $\delta$ from $\psi$.  The next lemma
asserts that if these bipartite states have some additional structure, then we
can find such a  unitary automorphism
that is similarly structured.

\begin{lemma}
\label{controlledlemma}
Suppose that $S, S', Q, R$ are registers, where $S$ is a copy of $S'$, and that $\psi, \phi$ are
pure states on $SS' QR$ that are supported on $\textnormal{Span} \{ e \otimes e \} \otimes Q \otimes R$,
where $e$ varies over the standard basis elements of $S, S'$.  Let $\delta = P ( \psi^{SQ}, \phi^{SQ} )$.  Then, there
exists an $S'$-controlled unitary operator $U$ on $S' \otimes R$ such that $\Delta ( U \psi , \phi ) = \delta$.  
\end{lemma}

\begin{proof}
Write $\psi = u u^*, \phi = v v^*$ with
\begin{eqnarray}
u & = & \sum_{e,f,g} \left( m_{fg}^e \right) e \otimes e \otimes f \otimes g \\
v & = & \sum_{e, f, g}\left(  n_{fg}^e \right) e \otimes e \otimes f \otimes g,
\end{eqnarray}
where $e, f, g$ vary over the standard basis elements of $S, Q, R$, respectively.  The fidelity $F(\psi^{SQ} , \phi^{SQ})$ is then given by the expression
\begin{eqnarray}
\sum_e \left\| (M^e)^* (N^e) \right\|_1,
\end{eqnarray}
where $M^e = [m_{fg}^e]_{fg}$ and $N^e = [n_{fg}^e ]_{fg}$ denote linear operators
from $R$ to $Q$.  Find unitary operators
$U^e \colon R \to R$ such that 
\begin{eqnarray}
\Tr [ U^e (M^e)^* (N^e) ] = \left\| (M^e)^* (N^e) \right\|_1.
\end{eqnarray}
Then, the controlled operator $\sum_e e e^* \otimes U^e $ on $S' \otimes R$ satisfies the desired condition.
\end{proof}

Now we prove a proposition about states
that approximate the behavior of players in a free nonlocal game. 
(The statement of this proposition is based in particular on the statement of Lemma~4.3
in \cite{Chung:2015}.)

\begin{proposition}
\label{approxgameprop}
Let $X, X'$ denote $\mathcal{X}$-valued registers,
let $A$ denote an $\mathcal{A}$-valued register, and let $Q = Q^1 Q^2 \cdots Q^n$ denote
a $n$-partite register.  
Let $\psi$ be a pure state of $XX'QA$ given by $\psi = u u^*$,
\begin{eqnarray}
\label{psistate}
u = \sum_{x \in \mathcal{X}} \sqrt{ \mu ( x ) } \left| xx \right> \otimes u_x,
\end{eqnarray}
where $\mu$ is a probability distribution on $\mathcal{X}$ and each $u_x$ is a unit vector in $Q A$, and suppose that
\begin{eqnarray}
H ( X^k \mid X^{\widehat{k}} {X'}^{\widehat{k}} Q^{\widehat{k}} A^{\widehat{k}} )_\psi & \geq & \log \left| \mathcal{X}^k \right| - \delta
\end{eqnarray}
for all $k \in \{ 1, 2, \ldots, n \}$.
Then,
\begin{eqnarray}
\label{epsiineq}
\mathbb{E}_\psi \left[ L ( X, A ) \right] & \leq & w (G) +  O ( \sqrt{\delta} ).
\end{eqnarray}
\end{proposition}

\begin{proof}
\textbf{Case 1:} Assume  that $\delta = 0$.

Then, the state of $X^k (X {X'} Q A )^{\widehat{k}}$ is uniform on $X^k$.  Making use of
Lemma~\ref{controlledlemma}, we can find unitary automorphisms $U^k_{x^k \to y^k}$ on $Q^kA^k$ for any $x^k, y^k \in \mathcal{X}^k$ such that
\begin{eqnarray}
U^k_{x^k \to y^k} u_{x^1 x^2 \cdots x^k \cdots x^n} = u_{x^1 x^2 \cdots y^k \cdots x^n}.
\end{eqnarray}
The expected score $\mathbb{E}_\psi \left[ L ( X, A ) \right]$ can be achieved at the game $G$ by having the $n$-players share some state of the form $u_{x} u_x^*$ with
$x \in \mathcal{X}$, receiving
an input sequence $y^1 \ldots y^n \in \mathcal{X}$, each applying the unitary $U^k_{x^k \to y^k}$ to their subsystem, and
then measuring $A^k$ to determine their output.  This is a valid quantum strategy, and so $\mathbb{E}_\psi \left[ L ( X, A ) \right]$
cannot exceed $w ( G )$.

\vskip0.15in

\textbf{Case 2:} General case.

Note that 
\begin{eqnarray}
I ( X^k : (XX'QA)^{\widehat{k}} )_\psi & \leq & \delta,
\end{eqnarray}
or equivalently,
\begin{eqnarray}
S ( \psi^{ X^k (XX'QA)^{\widehat{k}} } \| \psi^{ X^k} \otimes \psi^{ (XX'QA)^{\widehat{k}} } ) & \leq & \delta.
\end{eqnarray}
Therefore
\begin{eqnarray}
P ( \psi^{ X^k (XX'QA)^{\widehat{k}} } , \psi^{ X^k} \otimes \psi^{ (XX'QA)^{\widehat{k}} } ) & \leq & O ( \sqrt{ \delta } ).
\end{eqnarray}
Also, since $H ( X^k  )_\psi \geq \log \left| \mathcal{X}^k \right| - \delta$ and $I ( X^k : X^{1 \ldots (k-1) } ) \leq \delta$, the chain
rule implies $H ( X )_\psi \geq \log \left| \mathcal{X} \right| - O ( \delta ) $, and
therefore the distribution of $\mu$ is within purified distance $O ( \sqrt{ \delta } )$ from a uniform distribution.  Thus,
\begin{eqnarray}
\label{purifieddecouple}
P ( \psi^{ X^k (XX'QA)^{\widehat{k}} } , U_{ X^k} \otimes \psi^{ (XX'QA)^{\widehat{k}} } ) & \leq & O ( \sqrt{ \delta } ).
\end{eqnarray}

We will reduce to Case 1 via the use of a ``decoupling'' procedure.  Let $Y, Y'$ denote $\mathcal{X}$-valued registers. 
Let $\Psi$ be the state of $XX'YY'AQ$ such that $XX'AQ$ are in state $\psi$ and each register $Y^k {Y'}^k$ is in a Bell state.
Consider the following two-step process carried out on the state $\Psi$ by player $k$.
For simplicity, let $\Player^k = (XX'YY'AQ)^k$.  
\begin{enumerate}
\item \textbf{(Swap.)} Swap the state of the registers $X^k {X'}^k$ with the state of the registers $Y^k {Y'}^k$.

\item \textbf{(Recover.)} The state of the registers $X^k \Player^{\widehat{k}}$ is now
\begin{eqnarray}
\left( U_{ X^k} \right)  \otimes \left( \Psi^{ \Player^{\widehat{k}} } \right).
\end{eqnarray}
Using inequality (\ref{purifieddecouple}) and Lemma~\ref{controlledlemma}, apply an ${X'}^k$-controlled unitary operator $V^k$ to
the register $(X'YY'AQ)^k$ to bring the registers $\Player^{1 \ldots n}$ to a state that is within purified distance $O ( \sqrt{\delta })$
from state $\Psi$.
\end{enumerate}

Denote this process (which takes place on the registers $\Player^k$) by the symbol $\mathcal{U}^k$.  The state
$\mathcal{U}^k ( \Psi )$ is within purified distance $O ( \sqrt{ \delta } )$ from $\Psi$.
At the same time --- since the registers $X^k \Player^{\widehat{k}}$ are not used in step 2
--- we have $H ( X^k \mid \Player^{\widehat{k}} ) = \log \left| \mathcal{X} \right|$ under the state $\mathcal{U}^k ( \Psi )$. 

Applying the data processing inequality
and the triangle inequality, the state
\begin{eqnarray}
\mathcal{U}^1 \circ \mathcal{U}^2 \circ \cdots \circ \mathcal{U}^n ( \Psi ).  
\end{eqnarray}
is within $\Delta$-distance $O ( \sqrt{ \delta } )$ from $\Psi$,
and it also satisfies
\begin{eqnarray}
H ( X^k : \Player^{\widehat{k}} ) = \log \left| \mathcal{X} \right|
\end{eqnarray}
for all $k$.  The desired result therefore follows from Case 1.
\end{proof}

\subsection{The Pure Parallel Repetition Process}

\def\PureParallel{\textnormal{PureParallel}}

\begin{figure}[h]
\fbox{\parbox{6in}{\textbf{The Pure Parallel Repetition Process ($\PureParallel$)} \\ \\
\textit{Participants:} Players $1, \ldots, n$ and a referee.\\ \\ 
\textit{Parameters:} \\
\begin{tabular}{rl}
$N$: & A positive integer \\
$G$: & An $n$-player free game with input alphabet $\mathcal{X}$ and \\ & output alphabet $\mathcal{A}$ \\
$\Phi \in \mathcal{D} ( C^1 \otimes \ldots \otimes C^n ) $: & A pure $n$-partite state \\
$\{ M_{\mathbf{a} \mid \mathbf{x}} = M^1_{{\mathbf{a}^1} \mid {\mathbf{x}^1}} \otimes \ldots \otimes M^n_{{\mathbf{a}^n} \mid {\mathbf{x}^n}} \}_{\mathbf{a}, \mathbf{x}}$:
& A projective measurement strategy for $G^N$ \\ & (where $M^i_{\mathbf{a}^i \mid \mathbf{x}^i} \in \mathcal{P} ( C^i )$). 
\end{tabular} 
\\ \\
\textit{Registers:} \\
\begin{tabular}{rl}
$\{ A_k^i \}$: & Quantum registers (where $A_k^i$ is $\mathcal{A}^i$-valued, $1 \leq i \leq n$, $1 \leq k \leq N$) \\
$\{ X_k^i \}, \{ {X'}_k^i \}$: & Quantum registers (where ${X}_k^i, {X'}_k^i $ are $\mathcal{X}^i$-valued, $1 \leq i \leq n$, $1 \leq k \leq N$) \\
$\{ W_k \}$: & Bit registers ($1 \leq k \leq N$) \\
$P$ : & A $\Perm ( \{ 1, 2, \ldots, N \} )$-valued register.
\end{tabular} 
\\ \\ 
\textit{Procedure:} 
\begin{center} \parbox{5.5in}{
\begin{enumerate}
\item Players $1, 2, \ldots, N$ collectively prepare the register
$C$ in the state $\Phi$.

\item For each $i  \in \{ 1, 2, \ldots, N \}$, the $i$th player prepares the registers $\{ X_k^i \}_k$ and $\{ {X'}_k^i \}_k$
so that $X_k^i {X'}_k^i$ is in a Bell state for all $k$.

\item For each $i \in \{ 1, 2, \ldots, N \}$, the $i$th player applies the process from $X^iC^i$ to $X^i C^i A^i$ given by the unitary map
\begin{eqnarray*}
\left| \mathbf{x}^i \right> \left| v \right> & \mapsto & \sum_{\mathbf{a}^i \in \mathcal{A}^{\times N}} \left| \mathbf{x}^i \right>
\left| M^i_{\mathbf{a}^i \mid \mathbf{x}^i} v  \right> \left| \mathbf{a}^i \right>.
\end{eqnarray*}

\item The referee chooses a random permutation $\sigma \colon \{ 1, 2 , \ldots, N \} \to \{ 1, 2, \ldots, N \}$, stores
it in $P$ and broadcasts it to the players.  Each player applies the permutation $\sigma$ 
to their  registers $X^i_1, \ldots, X^i_N$, the registers ${X'}^i_1, \ldots, {X'}^i_N$, and the registers $A^i_1, \ldots, A^i_N$.

\item For each $i$, player $i$ measures $X_1^i, A_1^i$ and announces their values.  The referee sets
$W_1 := L ( X_1, A_1 )$.

\item For each $i$, player $i$ measures $X_2^i, A_2^i$ and announces their values.  The referee sets
$W_2 := L ( X_2, A_2 )$.

\item[] \hskip0.5in \vdots

\item[$N+4.$]  For each $i$, player $i$ measures $X_N^i, A_N^i$ and announces their values.  The referee sets
$W_N := L ( X_N, A_N )$.
\end{enumerate}
}\end{center} }}
\caption{A parallel repetition process with entangled inputs and a pure initial state.}
\label{parrepproc}
\end{figure}

We study the parallel repetition process given in Figure~\ref{parrepproc} (\PureParallel).  This $\PureParallel$ process
is similar to the process $\Par$ in Figure \ref{parfig}, except that it assumes the strategy used by the players involves
a pure state and projective measurements, and that they obtain their input symbols from a maximally entangled state.

For each $i \in \{ 1, 2, \ldots, n \}$ and $t \in \{ 1, 2, \ldots, N \}$, let $\Player^i_t$ denote the registers
of which player $i$ has knowledge at the conclusion of step $(t+4)$:
\begin{eqnarray}
\Player^i_t & = & X^i {X'}^i A^i C^i X^{\widehat{i}}_{1 \ldots t} A^{\widehat{i}}_{1 \ldots t} P.
\end{eqnarray}
Then, following our convention, $\Player^{\widehat{i}}_t$ denotes the registers of which players
$1, 2, \ldots, i-1, i+1, \ldots, n$ have knowledge at conclusion of step $(t+4)$:
\begin{eqnarray}
\Player^{\widehat{i}}_t & = & X^{\widehat{i}} {X'}^{\widehat{i}} A^{\widehat{i}} C^{\widehat{i}} X^{i}_{1 \ldots t} A^{i}_{1 \ldots t} P.
\end{eqnarray}

The next proposition asserts that, if the probability of winning the first $t$ rounds is not too unlikely, then these players
possess only a limited amount of information about player $i$'s input on the $(t+1)$st round.  The lower bound that we choose
for the winning probability in the first $t$ rounds can be somewhat arbitrary; we will take it to be $w(G)^{2t}$. 
\begin{proposition}
\label{infopossprop}
Suppose that $\mathbf{P} ( W_{1 \ldots t} = \mathbf{1}) \geq w ( G )^{2t}$ in $\PureParallel$.  Then, for any $i \in \{ 1,2, \ldots, n \}$,
the state $\Gamma^{t+4}$ that occurs after step $t+4$ satisfies
\begin{eqnarray}
H ( X^i_{t+1} \mid \Player_t^{\widehat{i}}, W_{1 \ldots t} = \mathbf{1}) & \geq & \log \left| \mathcal{X}^i \right| - O ( t/N).
\end{eqnarray}
\end{proposition}

\begin{proof}
Note that in the state $\Gamma^{t+4}$, the registers $X^i$ are uniformly distributed relative to $X^{\widehat{i}} {X'}^{\widehat{i}} A^{\widehat{i}} C^{\widehat{i}} P$.  Since
the conditional state $\Gamma^{t+4}_{\mid W_{1 \ldots t} = \mathbf{1}}$ satisfies
\begin{eqnarray}
(w( G ))^{2t} \cdot \Gamma^{t+4}_{\mid W_{1 \ldots t} = \mathbf{1} } \leq \Gamma^{t+4},
\end{eqnarray}
we have by Lemma~\ref{subuniformlemma} that
\begin{eqnarray}
H ( X^i \mid X^{\widehat{i}} {X'}^{\widehat{i}} A^{\widehat{i}} C^{\widehat{i}} P, W_{1 \ldots t} = \mathbf{1}) & \geq & N \log \left| \mathcal{X}^i \right| - O ( t ).
\end{eqnarray}
The registers $X^{i}_{1 \ldots t } A^i_{1 \ldots t}$ have a range of size $2^{O ( t ) }$, and so when we additionally  condition
on them we obtain
\begin{eqnarray}
H ( X^i_{(t+1) \ldots N} \mid \Player_t^{\widehat{i}}, W_{1 \ldots t} = \mathbf{1}) & \geq & N \log \left| \mathcal{X}^i \right| - O ( t ),
\end{eqnarray}
which implies
\begin{eqnarray}
\label{unifsum}
\sum_{j = t+1}^N H ( X^i_j \mid \Player_t^{\widehat{i}}, W_{1 \ldots t} = \mathbf{1}) & \geq & N \log \left| \mathcal{X}^i \right| - O ( t ).
\end{eqnarray}
By permutation symmetry, the value of every term in the summation in (\ref{unifsum}) is the same.\footnote{The permutation symmetry argument can be
made explicit as follows.
Let $p^j_\sigma = \mathbf{P} ( W_{1 \ldots t} = \mathbf{1} , P = \sigma )$.  Let 
 $s^j_\sigma := H ( X^i_j \mid \Player^{\widehat{i}},  W_{1 \ldots t} = \mathbf{1},  P = \sigma)$
if $p^j_\sigma \neq 0$ (and otherwise, let $s^j_\sigma = 0$). Then, the terms of the summation in (\ref{unifsum}) are the quantities $\left( \sum_\sigma p^j_\sigma s^j_\sigma \right)$ for 
$j \in \{ t+1, \ldots, N \}$.  For any $j, \ell \in \{ t+1, \ldots, N \}$, if we choose an $N$-permutation $\alpha$ that maps
$j$ to $\ell$ and fixes $\{ 1, 2, \ldots , t \}$, then $p^j_\sigma s^j_{\sigma} = p^\ell_{(\alpha \circ \sigma)} s^\ell_{(\alpha \circ \sigma)}$,
and so the quantities $\sum_\sigma p^j_\sigma s^j_\sigma$ and $\sum_\sigma p^\ell_\sigma s^\ell_\sigma$ are the same.}  Therefore,
\begin{eqnarray}
H ( X^i_{t+1} \mid \Player_t^{\widehat{i}}, W_{1 \ldots t} = \mathbf{1}) & \geq & \left[ N \log \left| \mathcal{X}^i \right|  - O ( t ) \right]/(N-t) \\
& \geq & \log \left| \mathcal{X}^i \right| - O ( t/N),
\end{eqnarray}
as desired.  
\end{proof}

We will use the previous proposition to prove by induction an upper bound on the probability that $W_{1 \ldots t} = \mathbf{1}$.

\begin{proposition}
\label{infopossprop2}
Suppose that $\mathbf{P} ( WIN ( t )) \geq w ( G )^{2t}$.  Then,
\begin{eqnarray}
\label{squarerooterror}
\mathbf{P} ( W_{1 \ldots (t+1)} = \mathbf{1} ) \leq \mathbf{P} ( W_{1 \ldots t} = \mathbf{1} ) \cdot (w ( G ) + O ( \sqrt{t / N  } ) ).
\end{eqnarray}
\end{proposition}

\begin{proof}
Consider the state of the $\PureParallel$ protocol after step $t+4$.
By Proposition~\ref{infopossprop}, the expected value of the quantity
\begin{eqnarray}
H ( X^i_{t+1} \mid \Player_t^{\widehat{i}}, X_{1 \ldots t}=x_{1 \ldots t }, A_{1 \ldots t }=a_{1 \ldots t},  P = \sigma,  W_{1 \ldots t} = \mathbf{1}),
\end{eqnarray}
when $x_{1 \ldots t }, a_{1 \ldots t}, \sigma$ vary according to the distribution given by the state $\Gamma^{t+4}_{\mid W_{1 \ldots t } = \mathbf{1} }$, is lower bounded
by $\log \left| \mathcal{X}^i \right| - O ( t/N)$.  Additionally, the state of the registers $\Player^{1 \ldots n }_t$ when conditioned on any such values
$X_{1\ldots t} = x_{1 \ldots t}, A_{1 \ldots t} = a_{1 \ldots t}, P = \sigma$, is a pure state.
By Proposition~\ref{approxgameprop} and the concavity of the square root function, the probability of the players winning the $(t+1)$st game under the distribution
$\Gamma^{t+4}_{\mid WIN ( t ) }$ is no more than $w ( G ) + O ( \sqrt { t / N } )$, as desired.
\end{proof}

\begin{theorem}
\label{parallelthm}
For any $t \in \{ 1, 2, \ldots, N \}$,
\begin{eqnarray}
\label{shortroundineq}
\mathbf{P} ( W_{1 \ldots t} = \mathbf{1} ) \leq (w ( G ) + O ( \sqrt{t / N } ))^t.
\end{eqnarray}
\end{theorem}

\begin{proof}
Let $E$ denote the function represented by $O$ on the righthand side of inequality (\ref{squarerooterror}).
We apply induction on $t$.  The base case is obvious.  For the inductive step, assume that
\begin{eqnarray}
\mathbf{P} ( W_{1 \ldots t} = \mathbf{1} ) & \leq & (w ( G ) + E ( \sqrt{ t/N }  ) )^t
\end{eqnarray}
holds
for a given value of $t \in \{ 1, 2, \ldots, N-1 \}$.  If $\mathbf{P} ( W_{1 \ldots t} = \mathbf{1} ) < (w ( G ) )^{2t}$, then 
\begin{eqnarray}
\mathbf{P} ( W_{1 \ldots (t+1)} = \mathbf{1} ) & < & (w ( G ))^{2t} \\ & \leq & (w ( G ))^{t+1}, \end{eqnarray} and there is nothing to prove.  If 
$\mathbf{P} ( W_{1 \ldots t} = \mathbf{1} ) \geq (w ( G ) )^{2t}$, then by Proposition~\ref{infopossprop2}, 
\begin{eqnarray}
\label{inddeduc}
\mathbf{P} ( W_{1 \ldots (t+1)} = \mathbf{1} ) \leq (w ( G ) + E ( \sqrt{t / N } ))^{t+1},
\end{eqnarray}
which completes the proof.
\end{proof}

\clearpage
\bibliographystyle{plain}

\bibliography{../resources/ParallelDIQKD}

\end{document}